\documentclass[11pt]{amsart}
\usepackage{amssymb}
\usepackage{amsfonts}

\newtheorem{theorem}{Theorem}
\theoremstyle{plain}
\newtheorem{acknowledgement}{Acknowledgement}

\newtheorem{corollary}{Corollary}

\newtheorem{definition}{Definition}

\newtheorem{lemma}{Lemma}
\newtheorem{notation}{Notation}

\newtheorem{proposition}{Proposition}
\newtheorem{remark}{Remark}

\numberwithin{equation}{section}
\input{tcilatex}

\begin{document}
\title[Polar Duality on Symplectic Spaces]{Symplectic and Lagrangian Polar
Duality; Applications to Quantum Information Geometry}
\author{Maurice de Gosson}
\address{University of Vienna\\
Institute of mathematics (NuHAG)}
\email{maurice.de.gosson@univie.ac.at}
\urladdr{https://cvdegosson.webs.com/}
\author{Charlyne de Gosson}
\date{2023}
\thanks{}
\dedicatory{Dedicated to Leonid Polterovich on his 60th birthday}

\begin{abstract}
Polar duality is a well-known concept from convex geometry and analysis. In
the present paper we study two symplectically covariant versions of polar
duality having in mind their applications to quantum mechanics. The first
variant makes use of the symplectic form on phase space and allows a precise
study of the covariance matrix of a density operator. The latter is a
fundamental object in quantum information theory., The second variant is a
symplectically covariant version of usual polar duality highlighting the
role played by Lagrangian planes. It allows us to define a the notion of
\textquotedblleft geometric quantum states\textquotedblleft\ with are in
bijection with generalized Gaussians.
\end{abstract}

\maketitle
\tableofcontents

\section{Introduction}

The concept of polar dual set in convex geometry corresponds to the concept
of dual space in linear algebra. Given a convex body $X$ in the Euclidean
space $\mathbb{R}^{n}$ its polar dual is the set $X^{\hbar }$ of all $p\in (%
\mathbb{R}^{n})^{\ast }$ such that $\langle p,x\rangle \leq \hbar $; here $%
\hbar $ is a positive constant, usually taken to be one in the standard
literature (we use for flexibility a parameter-dependent definition; in
quantum mechanics $\hbar $ would be Planck's constant $h$ divided by $2\pi $%
; in harmonic analysis one would take $\hbar =1/2\pi $ while the standard
choice in the theory of partial differential equations is $\hbar =1$). We
will most of the time identify the dual space $(\mathbb{R}^{n})^{\ast }$
with $\mathbb{R}^{n}$ itself, in which case the polar dual of $X$ is
identified, using the standard Euclidean structure $(x,p)\longmapsto p\cdot
x $ with the set 
\begin{equation}
X^{\hbar }=\{p\in \mathbb{R}^{n}:\sup\nolimits_{x\in X}(p\cdot x)\leq \hbar
\}.  \label{polar1}
\end{equation}

We will consider two variants of polar duality on the symplectic space $%
(T^{\ast }\mathbb{R}^{n},\sigma )$ where $\sigma $ is the standard
symplectic form $\sum_{j=1}^{n}dp_{j}\wedge dx_{j}$; we will identify $%
T^{\ast }\mathbb{R}^{n}$ with $\mathbb{R}^{2n}$ and from time to time use
the suggestive notation $\mathbb{R}^{2n}=\mathbb{R}_{x}^{n}\times \mathbb{R}%
_{p}^{n}$.

The first variant is what we call \textquotedblleft symplectic polar
duality\textquotedblright : if $\Omega \subset \mathbb{R}^{2n}$ is a convex
body we define its symplectic polar dual by%
\begin{equation}
\Omega ^{\hbar ,\sigma }=\{z^{\prime }\in \mathbb{R}^{2n}:\sup\nolimits_{z%
\in \Omega }\sigma (z,z^{\prime })\leq \hbar \};  \label{polar2}
\end{equation}%
clearly $\Omega ^{\hbar ,\sigma }=J(\Omega ^{\hbar })$ where $J$ is the
standard symplectic automorphism $J(x,p)=(p,-x)$. The interest of this
notion comes (among other things) from the fact that it has the symplectic
covariance property $S(\Omega ^{\hbar ,\sigma })=(S(\Omega ))^{\hbar ,\sigma
}$ for every $S\in \limfunc{Sp}(n)$.

The second variant, which we will refer to as \textquotedblleft Lagrangian
polar duality\textquotedblright ,\ is of a slightly more subtle nature. Let $%
(\ell ,\ell ^{\prime })$ be a pair of transverse Lagrangian planes in $(%
\mathbb{R}^{2n},\sigma )$. If $X_{\ell }$ is a convex body contained in $%
\ell $ then its Lagrangian polar dual $(X_{\ell })_{\ell ^{\prime
}}^{\hslash }$ with respect to $\ell ^{\prime }$ is, by definition,%
\begin{equation}
(X_{\ell })_{\ell ^{\prime }}^{\hbar }=\{z^{\prime }\in \ell ^{\prime
}:\sup\nolimits_{z\in \ell }\sigma (z,z^{\prime })\leq \hbar \}.
\label{polar3}
\end{equation}%
When $\ell =\mathbb{R}^{n}\times 0$ and $\ell ^{\prime }=0\times \mathbb{R}%
^{n}$ (the canonical coordinate Lagrangian planes) Lagrangian polar duality
reduces to ordinary polar duality on $\mathbb{R}^{n}$ (\ref{polar1}).

As witnessed by our choice of parameter $\hbar $, we have in mind the
Science of quantum mechanics when defining these new notions of polar
duality. As we will see in the course of this paper, symplectic polar
duality is closely related to difficult questions of positivity for trace
class operators, and allows to express \textquotedblleft quantization
conditions\textquotedblright\ in an elegant an concise geometric way. On the
other hand, Lagrangian polar duality allows a geometric redefinition of the
notion of quantum state; these states are usually viewed as
\textquotedblleft wavefunctions\textquotedblright\ in physics; in our
approach they appear as geometric objects defined in terms of convex
products $X_{\ell }\times (X_{\ell })_{\ell ^{\prime }}^{\hbar }$ whose
factors are supported by transversal Lagrangian planes, and their functional
aspects appear only as subsidiary through the use of the John ellipsoid
(maximum volume ellipsoid).

The applications of concepts of convex geometry and analysis outside their
original area is not new, see for instance Milman \cite{Milman} who applies
such methods to probability theory; also see the treatise \cite{ABMB} by
Aubrun and Szarek.

Let us describe some highlights of this work, emphasizing what we hold for
the most important results (our choice being of course somewhat subjective,
and highly depending on the authors' tastes). The paper consists of two
parts, which can be read independently:

\subsubsection*{Part 1: Symplectic polar duality and the covariance matrix}

The covariance matrix of a physical state (be it classical, or quantum) is a
statistical object whose importance in information theory is crucial; it
encodes the statistical properties of the state and its study is, as we will
see, greatly facilitated by the polar duality approach. The central result
is, no doubt, Theorem \ref{Thm2} who gives two criteria for what we call
\textquotedblleft quantum admissibility\textquotedblright\ of a phase space
ellipsoid (the definition of this notion of admissibility is closely related
to the uncertainty principle of quantum mechanics, an is rigorously defined
in Definition \ref{defadm}). The first criterion say that an ellipsoid $%
\Omega $ is admissible if and only if it contains its symplectic polar dual $%
\Omega ^{\hbar ,\sigma }$; the second is of a more subtle nature; it shows
that it is sufficient (and necessary) for admissibility that $\Omega ^{\hbar
,\sigma }\cap F\subset \Omega \cap F$ for every symplectic subspace of $(%
\mathbb{R}^{2n},\sigma )$. This is a tomographic condition reminiscent of an
old result by Narcowich \cite{Narcow} concerning covariance and information
ellipsoid in quantum mechanics. This analogy is made even more convincing in
Theorem \ref{ThNarcow} where we give a dynamical description of quantum
admissibility of covariance and information ellipsoids using the techniques
we develop; in particular the role of the so fruitful notion of symplectic
capacity is highlighted (symplectic capacities are strongly related to
Gromov's famous non-squeezing theorem). We take this opportunity to give a
new functional-analytical characterization of those (classical, or quantum)
state for which the covariance matrix is well-defined. This is done in terms
of a class of modulation spaces.

\subsubsection*{Part 2: Lagrangian Polar Duality and Geometric Quantum States%
}

The main objects we study are here geometric quantum states in $\mathbb{R}%
^{2n}$ associated with a Lagrangian frame $(\ell ,\ell ^{\prime })$. By
definition such a state is a Cartesian product $X_{\ell }\times (X_{\ell
})_{\ell ^{\prime }}^{\hbar }$ where $X_{\ell }$ is an ellipsoid carried by $%
\ell $ and $(X_{\ell })_{\ell ^{\prime }}^{\hbar }$ its Lagrangian dual in $%
\ell ^{\prime }$. The set of all such products is denoted by $\limfunc{Quant}%
\nolimits_{0}^{\mathrm{Ell}}(n)$; we show that there is a natural transitive
action of the symplectic group on $\limfunc{Quant}\nolimits_{0}^{\mathrm{Ell}%
}(n)$. The application of these constructions to traditional Gaussian
quantum mechanics is given i Theorem \ref{Thm1} which says that one can
identify the set $\limfunc{Gauss}\nolimits_{0}(n)$ of centered generalized
Gaussians $\psi _{A,B}^{\gamma }=e^{i\gamma }\psi _{A,B}$ with 
\begin{equation}
\psi _{A,B}(x)=\left( \tfrac{1}{\pi \hbar }\right) ^{n/4}(\det A)^{1/4}e^{-%
\tfrac{1}{2\hbar }A+iB)x\cdot x}
\end{equation}%
($A,B\in \limfunc{Sym}(n,\mathbb{R})$, $A>0$, $\gamma \in \mathbb{R}$) with $%
\limfunc{Quant}\nolimits_{0}^{\mathrm{Ell}}(n)$). We thereafter define the
set $\limfunc{Quant}\nolimits^{\mathrm{Ell}}(n)$ of geometric states with
arbitrary center, and in Theorem \ref{ThmBeam} we study the propagation of
geometric states in $\limfunc{Quant}\nolimits^{\mathrm{Ell}}(n)$ under the
action of first order \textquotedblleft Gaussian beams\textquotedblright\
which are approximations to the Hamiltonian motion for a wide class of
Hamiltonian functions.

\begin{notation}
The standard symplectic form $\sigma $ is written in matrix form as $\sigma
(z,z^{\prime })=Jz\cdot z^{\prime }$ where $J=%
\begin{pmatrix}
0_{n\times n} & I_{n\times n} \\ 
-I_{n\times n} & 0_{n\times n}%
\end{pmatrix}%
$ is the standard symplectic matrix. The standard symplectic group $\limfunc{%
Sp}(n)$ is the group of automorphisms $S$ of $T^{\ast }\mathbb{R}^{n}\equiv 
\mathbb{R}^{2n}$ such that $S^{\ast }\sigma =\sigma $; in matrix notation: $%
S\in \limfunc{Sp}(n)$ if and only if $STJS=J$ (or, equivalently, $SJS^{T}=J$%
), $S^{T}$ the transpose of $S$). The unitary representation of the double
cover of $\limfunc{Sp}(n)$ (the metaplectic group) is denoted by $\pi ^{%
\limfunc{Mp}}:\limfunc{Mp}(n)\longrightarrow \limfunc{Sp}(n)$. The
Lagrangian Grassmannian of $(\mathbb{R}^{2n},\sigma )$ is denoted $\limfunc{%
Lag}(n)$, thus $\ell \in \limfunc{Lag}(n)$ if and only if $\ell $ is a
linear subspace of $\mathbb{R}^{2n}$ with $\dim \ell =n$ and $\sigma |\ell =0
$. The elements of $\limfunc{Lag}(n)$ will be called \textit{Lagrangian
planes}.
\end{notation}

\part{Symplectic Polar Duality and the Covariance Matrix}

\section{Definition and elementary properties}

\subsection{A short review of usual polar duality}

For detailed treatments of the topics of convex geometry and analysis used
in this article we refer to the treatise \cite{ABMB} by Aubrun and. Szarek
and to Vershynin's online lecture notes \cite{Vershynin}. For a
comprehensive study of convex geometry with applications to optimization
theory see Boyd \textit{et al.} \cite{Boyd}.

\subsubsection{Definition}

Let $X$ and $Y$ be convex sets in \ $\mathbb{R}^{n}$; then, 
\begin{gather}
(X\cup Y)^{\hbar }=X^{\hbar }\cap Y^{\hbar }\text{ \ , \ }(X\cap Y)^{\hbar }=%
\widetilde{X^{\hbar }\cup Y^{\hbar }}  \label{prop1} \\
X\subset Y\Longrightarrow Y^{\hbar }\subset X^{\hbar }\text{ \ , \ }X\text{ 
\textit{closed} }\Longrightarrow (X^{\hbar })^{\hbar }=X  \label{prop2} \\
A\in GL(n,\mathbb{R})\Longrightarrow (AX)^{\hbar }=(A^{T})^{-1}X^{\hbar }
\label{prop3}
\end{gather}%
(in the second formula (\ref{prop1}) $\widetilde{X^{\hbar }\cup Y^{\hbar }}$
is the convex hull of $X^{\hbar }\cup Y^{\hbar }$). If $A=A^{T}\in GL(n,%
\mathbb{R})$ is positive definite then 
\begin{equation}
\{x\in \mathbb{R}_{x}^{n}:Ax\cdot x\leq \hbar \}^{\hbar }=\{p\in \mathbb{R}%
_{p}^{n}:A^{-1}p\cdot p\cdot \leq \hbar \}  \label{ell}
\end{equation}%
hence, in particular, 
\begin{equation}
B_{X}^{n}(\sqrt{\hbar })^{\hbar }=B_{P}^{n}(\sqrt{\hbar })\text{ },\text{ }%
B_{P}^{n}(\sqrt{\hbar })^{\hbar }=B_{X}^{n}(\sqrt{\hbar })  \label{balls}
\end{equation}%
where $B_{X}^{n}(\sqrt{\hbar })$ (\textit{resp}. $B_{P}^{n}(\sqrt{\hbar })$)
is the ball in $\mathbb{R}_{x}^{n}$ (\textit{resp}. $\mathbb{R}_{p}^{n}$)
with radius $\sqrt{\hbar }$ and centered at $0$.

\begin{remark}
The properties of polar duality is less transparent for convex bodies not
centered at the origin and requires the use of the so-called Santal\'{o}
point \cite{Santalo}.
\end{remark}

\subsubsection{Projections and intersections}

Polar duality exchanges the projection and the intersection operations \cite%
{ABMB}. While this result seems to be well-known it seems difficult to find
a detailed proof in the literature, so we prove this important result,
following \cite{Vershynin}. For this we need the following elementary lemma:

\begin{lemma}
\label{propell}\textit{Let} $X=\{x:Ax\cdot x\leq 1\}$ \textit{and }$%
P=\{p:Bp\cdot p\leq 1\}$ ($A,B$ positive definite and symmetric) be two
ellipsoids. We have $X^{\hbar }\subset P$ \textit{if and only if }$A\leq
B^{-1}$, \textit{and }$X^{\hbar }=P$ \textit{if and only if} $AB=I_{n\times
n}$.
\end{lemma}

\begin{proof}
We have $X=A^{-1/2}(B_{X}^{n}(\sqrt{\hbar }))$ and $P=B^{-1/2}(B_{P}^{n}(%
\sqrt{\hbar }))$ and\ the inclusion$X^{\hbar }\subset P$ is thus equivalent
to the inequality $A^{1/2}\leq B^{-1/2}$ in the L\"{o}wner ordering, that
is, to $A\leq B^{-1}$ with equality if and only if $X^{\hbar }=P$.
\end{proof}

Polar duality exchanges the operations of intersection and orthogonal
projection:

\begin{proposition}
\label{Propinter}Let $X\subset $ $\mathbb{R}_{x}^{n}$ be a convex body
containing $0$ in its interior and $F$ a linear subspace of $\mathbb{R}%
_{x}^{n}$; we have%
\begin{equation}
(\Pi _{F}X)^{\hbar }=X^{\hbar }\cap F\text{ \ , \ }(X\cap F)^{\hbar }=\Pi
_{F}(X^{\hbar })  \label{projdual}
\end{equation}%
where $\Pi _{F}$ us the orthogonal projection in $\mathbb{R}_{x}^{n}$ onto $%
F $. [In $(\Pi _{F}X)^{\hbar }$ and $(X\cap F)^{\hbar }$ the polar duals are
taken inside the subspace $F$ equipped with the induced inner product.]
\end{proposition}

\begin{proof}
It suffices to prove the first formula (\ref{projdual}) since the second
follows by duality: 
\begin{equation*}
X\cap F=(X^{\hbar })^{\hbar }\cap F=(\Pi _{F}X^{\hbar })^{\hbar }
\end{equation*}%
and hence \ $(X\cap F)^{\hbar }=\Pi _{F}(X^{\hbar })$. Let us next show that 
$\Pi _{F}(X^{\hslash })\subset (X\cap F)^{\hbar }$. For $p\in X^{\hslash }$
we have, for every $x\in X\cap F$, 
\begin{equation*}
x\cdot \Pi _{F}p=\Pi _{F}x\cdot p=x\cdot p\leq \hbar
\end{equation*}%
hence $\Pi _{F}p\in (X\cap F)^{\hbar }$. To prove the inclusion $\Pi
_{F}(X^{\hslash })\supset (X\cap F)^{\hbar }$ we note that it is sufficient,
by the anti-monotonicity and reflexivity properties of polar duality, to
prove that $(\Pi _{F}(X^{\hslash }))^{\hbar }\subset X\cap F$. Let $x\in
(\Pi _{F}(X^{\hslash }))^{\hbar }$; we have $x\cdot \Pi _{F}p\leq \hbar $
for every $p\in X^{\hslash }$. Since $x\in F$ (because the dual of a subset
of $F$ is taken inside $F$) we also have 
\begin{equation*}
\hbar \geq x\cdot \Pi _{F}p=\Pi _{F}x\cdot p=x\cdot p
\end{equation*}%
from which follows that $x\in (X^{\hbar })^{\hbar }=X$, which shows that $%
x\in X\cap F$. This concludes the proof.
\end{proof}

\subsubsection{John and L\"{o}wner ellipsoids}

A fundamental tool in convex geometry and Banach space geometry is the John
ellipsoid of a convex body $\Omega \subset \mathbb{R}^{2n}.$ It is \cite%
{Ball,Boyd} the (unique) ellipsoid $\Omega _{\mathrm{John}}$ of maximal
volume contained in $\Omega $; similarly the (unique) minimum enclosing
ellipsoid is the L\"{o}wner ellipsoid $\Omega _{\mathrm{L\ddot{o}wner}}$. \
We note that both the John and the L\"{o}wner ellipsoids transform
covariantly under linear (and affine) transforms: if $L\in GL(n,\mathbb{R})$
then%
\begin{equation}
(L(X))_{\mathrm{John}}=L(X_{\mathrm{John}})\text{ \ \ },\text{ \ }(L(X))_{%
\mathrm{L\ddot{o}wner}}=L(X_{\mathrm{John}}).  \label{JLcov}
\end{equation}

\begin{remark}
\label{remchat}One can define the John ellipsoid in any finite-dimensional
normed space, regardless of whether the space is Euclidean or not. The
definition of the John ellipsoid in a normed space is the same as the one
given earlier.
\end{remark}

Polar duality interchanges the John and L\"{o}wner ellipsoids; we have the
following duality relations hold for convex symmetric bodies \cite{ABMB}:%
\begin{equation}
(X_{\mathrm{John}})^{\hbar }=(X^{\hbar })_{\mathrm{L\ddot{o}wner}}\text{ \ \ 
},\text{ \ }(X_{\mathrm{L\ddot{o}wner}})^{\hbar }=(X^{\hbar })_{\mathrm{John}%
}.  \label{JL}
\end{equation}

\subsubsection{Blaschke--Santal\'{o} inequality and Mahler volume}

Let $X$ be an origin symmetric convex body in $\mathbb{R}_{x}^{n}$. By
definition, the Mahler volume (or volume product) of $X$ is the product%
\begin{equation}
\upsilon (X)=\limfunc{Vol}\nolimits_{n}(X)\limfunc{Vol}\nolimits_{n}(X^{%
\hbar })  \label{Mahler}
\end{equation}%
where $\limfunc{Vol}\nolimits_{n}$ is the usual Euclidean volume on $\mathbb{%
R}_{x}^{n}$. The Mahler volume is a dimensionless quantity because of its
rescaling invariance: we have $\upsilon (\lambda X)=\upsilon (X)$ for all $%
\lambda >0$. More generally, the Mahler volume is invariant under linear
automorphisms of $\mathbb{R}_{x}^{n}$: if $L\in GL(n,\mathbb{R})$ we have%
\begin{eqnarray}
\upsilon (LX) &=&\limfunc{Vol}\nolimits_{n}(LX)\limfunc{Vol}%
\nolimits_{n}(L^{T})^{-1}X^{\hbar })  \label{mahlerinv} \\
&=&\limfunc{Vol}\nolimits_{n}(X)\limfunc{Vol}\nolimits_{n}(X^{\hbar })=v(X).
\end{eqnarray}%
A remarkable property of polar duality is the Blaschke--Santal\'{o}
inequality \cite{Blaschke}: assume again that $X$ is a centrally symmetric
body; then 
\begin{equation}
\upsilon (X).\leq (\limfunc{Vol}\nolimits_{n}(B^{n}(\sqrt{\hbar }))^{2}=%
\frac{(\pi \hbar )^{n}}{\Gamma (\frac{n}{2}+1)^{2}}  \label{sant1}
\end{equation}%
and equality is attained if and only if $X\subset \mathbb{R}_{x}^{n}$ is an
ellipsoid centered at the origin (see \cite{Bianchi} for a proof using
Fourier analysis). It is conjectured (the \textquotedblleft Mahler
conjecture\textquotedblright\ \cite{Mahler}) that one has the lower bound%
\begin{equation}
\upsilon (X)\geq \frac{(4\hbar )^{n}}{n!}  \label{volvo3}
\end{equation}%
with equality only when $X$ is the hypercube $C=[-1,1]^{n}$. Bourgain and
Milman \cite{BM} have shown the existence, for every $n\in \mathbb{N}$, of a
constant $C_{n}>0$ such that 
\begin{equation}
\limfunc{Vol}\nolimits_{n}(X)\limfunc{Vol}\nolimits_{n}(X^{\hbar })\geq
C_{n}\hbar ^{n}/n!  \label{BM}
\end{equation}%
and more recently Kuperberg \cite{Kuper} has shown that one can choose $%
C_{n}=(\pi /4)^{n}$.

\begin{remark}
In view of the invariance property (\ref{mahlerinv}) this is equivalent to
saying that the minimum is attained by any $n$-parallelepiped%
\begin{equation}
X=[-\sqrt{2\sigma _{x_{1}x_{1}}},\sqrt{2\sigma _{x_{1}x_{1}}}]\times \cdot
\cdot \cdot \times \lbrack -\sqrt{2\sigma _{x_{n}x_{n}}},\sqrt{2\sigma
_{x_{n}x_{n}}}].  \label{interval}
\end{equation}%
This is related to the covariances of the tensor product $\psi =\phi
_{1}\otimes \cdot \cdot \cdot \otimes \phi _{n}$ of standard one-dimensional
Gaussians $\phi _{j}(x)=(\pi \hbar )^{-1/4}e^{-x_{j}^{2}/2\hbar }$; the
function $\psi $ is a minimal uncertainty quantum state in the sense that it
reduces the Heisenberg inequalities to equalities. This observation might
lead to a \textquotedblleft quantum proof\textquotedblright\ of the Mahler
conjecture.
\end{remark}

\subsection{Example: Hardy's uncertainty principle}

Here is an elementary application of polar duality which highlights the role
it plays in questions related to the uncertainty principle of quantum
mechanics. We are following the presentation we gave in \cite{acha}.

Hardy's uncertainty principle \cite{Hardy} in its original (one dimensional)
formulation says that if the moduli of $\psi \in L^{1}(\mathbb{R})\cap L^{2}(%
\mathbb{R})$ and of its Fourier transform, here defined by 
\begin{equation*}
\widehat{\psi }(p)=F\psi (p)=\frac{1}{\sqrt{2\pi \hbar }}\int_{-\infty
}^{\infty }e^{-\frac{i}{\hbar }px}\psi (x)dx
\end{equation*}%
are different from zero and satisfy estimates%
\begin{equation*}
|\psi (x)|\leq C_{A}e^{-\frac{a}{2\hbar }x^{2}}\text{ \ , \ }|\widehat{\psi }%
(p)|\leq C_{B}e^{-\frac{b}{2\hbar }p^{2}}
\end{equation*}%
($C_{A}$, $C_{B}>0$, $a,b>0$), then we must have $ab\leq 1$ and if $ab=1$ we
have $\psi (x)=Ce^{-\frac{a}{2\hbar }x^{2}}$ for some complex constant $C.$
We have proven a multidimensional version of this result (\cite{golulett}, 
\cite{Wigner}, Chapter 10).Let $\psi ,\widehat{\psi }\in L^{1}(\mathbb{R}%
^{n})\cap L^{2}(\mathbb{R}^{n})$, $\psi \neq 0$ where 
\begin{equation*}
\widehat{\psi }(p)=F\psi (p)=\frac{1}{(2\pi \hbar )^{n/2}}\int_{\mathbb{R}%
^{n}}e^{-\frac{i}{\hbar }p\cdot x}\psi (x)dx.
\end{equation*}%
Then:

\begin{proposition}
\label{prophardy1}Let $A,B\in \limfunc{Sym}(n,\mathbb{R})$ be positive
definite and $\psi $ as above. Assume that there exist a constants $%
C_{A},C_{B}>0$ such that 
\begin{equation}
|\psi (x)|\leq C_{A}e^{-\tfrac{1}{2\hbar }Ax\cdot x}\text{ \ and \ }|F\psi
(p)|\leq CBe^{-\tfrac{1}{2\hbar }Bp\cdot p}.
\end{equation}

(i) The eigenvalues $\lambda _{j}$, $1\leq j\leq n$, of $AB$ are $\leq 1$;
(ii) If $\lambda _{j}=1$ for all $j$, then $\psi (x)=ke^{-\frac{1}{2\hbar }%
Ax^{2}}$ for some $k\in \mathbb{C}$.
\end{proposition}

$\int t$ turns out that this result can easily be restated in terms of polar
duality:

\begin{corollary}
The Hardy estimates are satisfied by a non-zero function $\psi \in L^{2}(%
\mathbb{R}^{n})$ if and only if the ellipsoids $X=\{x:Ax\cdot x\leq \hbar \}$
and $P=\{p:Bx\cdot x\leq \hbar \}$ satisfy $X^{\hbar }\subset P$ with
equality $X^{\hbar }=P$ if and only $\psi (x)=ke^{-\frac{1}{2\hbar }Ax\cdot
x}$ for some $k\in \mathbb{C}$.
\end{corollary}

\begin{proof}
This immediately follows from Lemma \ref{propell}.
\end{proof}

We will apply the results above to sub-Gaussian estimates of the Wigner
function in Section \ref{secsubgauss}.

\section{The Symplectic Case}

In what follows we identify the dual of the symplectic space $(\mathbb{R}%
^{2n},\sigma )$ with itself.

\subsection{Basic properties}

\subsubsection{Symplectic covariance}

Let $J=%
\begin{pmatrix}
0_{n\times n} & I_{n\times n} \\ 
-I_{n\times n} & 0_{n\times n}%
\end{pmatrix}%
$ be the standard symplectic matrix. Using the matrix formulation $\sigma
(z,z^{\prime })=Jz\cdot z^{\prime }$ of the symplectic form it is
straightforward to verify It is straightforward to verify that the
symplectic polar dual 
\begin{equation}
\Omega ^{\hbar ,\sigma }=\{z^{\prime }\in \mathbb{R}^{2n}:\sup\nolimits_{z%
\in \Omega }\sigma (z,z^{\prime })\leq \hbar \}  \label{sypodu}
\end{equation}%
of a convex body $\Omega \subset \mathbb{R}^{2n}$ is related to the ordinary
polar dual $\Omega ^{\hbar }$ by the formula 
\begin{equation}
\Omega ^{\hbar ,\sigma }=(J\Omega )^{\hbar }=J(\Omega ^{\hbar }).
\label{omegaj}
\end{equation}%
It follows that:

\begin{lemma}
\label{lemmaelldual}Let $\Omega _{M}$ be the phase space ellipsoid defined by%
$Mz\cdot z\leq r^{2}$ where $M$ is symmetric and positive definite and $r>0$%
. Then 
\begin{equation}
\Omega _{M}^{\hbar ,\sigma }=\{z:-JM^{-1}Jz\cdot z\leq (\hbar /r)^{2}\}.
\label{ohasig}
\end{equation}
\end{lemma}

\begin{proof}
The ordinary $\hbar $-polar dual $\Omega _{M}^{\hbar }$ of $\Omega _{M}$ is
defined by $M^{-1}z\cdot z\leq (\hbar /r)^{2}$. In view of (\ref{omegaj}) $%
\Omega _{M}^{\hbar ,\sigma }=J(\Omega _{M})$ hence the result.
\end{proof}

Formula (\ref{omegaj}) can easily be generalized to yield the following
important symplectic covariance result:

\begin{proposition}
\label{Prop1}Let $\Omega $ be a symmetric convex body and $S\in \limfunc{Sp}%
(n)$. (i) we have 
\begin{equation}
(S(\Omega ))^{\hbar ,\sigma }=S(\Omega ^{\hbar ,\sigma }).
\label{sympcodual}
\end{equation}%
(ii) More generally for $S\in \limfunc{Sp}(n)$ and $F$ a linear subspace of $%
\mathbb{R}^{2n}$ we have 
\begin{equation}
S(\Omega \cap F)^{\hbar ,\sigma }=(S\Omega \cap SF)^{\hbar ,\sigma }.
\label{symplectogonal}
\end{equation}
\end{proposition}

\begin{proof}
(i) Using successively (\ref{omegaj}), the scaling property (\ref{prop3}) in
dimension $2n$, and again (\ref{omegaj}), we have 
\begin{eqnarray*}
S(\Omega ^{\hbar ,\sigma }) &=&SJ(\Omega ^{\hbar })=J(S^{T})^{-1}(\Omega
^{\hbar }) \\
&=&J(S(\Omega ))^{\hbar }=(S(\Omega ))^{\hbar ,\sigma }.
\end{eqnarray*}%
(ii) Formula (\ref{symplectogonal}) follows from formula (\ref{sympcodual})
since $\Omega \cap F$ is convex and symmetric.
\end{proof}

\subsubsection{Quantum blobs}

Symplectic balls with radius $\sqrt{\hbar }$ are the only fixed ellipsoids
for ordinary polar duality. They play an important role in various
formulations of the uncertainty principle of quantum mechanics \cite%
{physletta,blob,golu09} where they represent minimum uncertainty units; this
motivates the following definition:

\begin{definition}
A quantum blob $Q_{S}(z_{0})$ is a symplectic ball with radius $\sqrt{\hbar }
$: $Q_{S}(z_{0})=S(B^{2n}(z_{0},\sqrt{\hbar }))$ for some $S\in \limfunc{Sp}%
(n)$. When $z_{0}=0$ we write $Q_{S}=Q_{S}(0)$.
\end{definition}

We will see later (Lemma \ref{Gromov}) that a characteristic property of
quantum blobs is that their orthogonal projections on symplectic planes can
never become smaller than $\pi \hbar $.

\begin{proposition}
\label{propfix}Let $\Omega $ be a centered ellipsoid in $(\mathbb{R}%
^{2n},\sigma )$. We have $\Omega =\Omega ^{\hbar ,\sigma }$ if and only if $%
\Omega $ is a quantum blob, i.e. if $\Omega =Q_{S}$ for some $S\in \limfunc{%
Sp}(n)$.
\end{proposition}

\begin{proof}
That $Q_{S}^{\hbar ,\sigma }=Q_{S}$ is clear in view of (\ref{sympcodual}):%
\begin{equation*}
Q_{S}^{\hbar ,\sigma }=S(B^{2n}(\sqrt{\hbar }))^{\hbar ,\sigma }=S(B^{2n}(%
\sqrt{\hbar }))=Q_{S}.
\end{equation*}%
Suppose conversely that the ellipsoid $\Omega $ is defined by $Mz\cdot z\leq
\hbar $; then its symplectic polar dual $\Omega ^{\hbar ,\sigma }$ is
defined by $-JM^{-1}Jz\cdot z\leq \hbar $ (Lemma \ref{lemmaelldual}) and we
have $\Omega =\Omega ^{\hbar ,\sigma }$ if and only if $M=-JM^{-1}J$. This
condition is trivially equivalent to $MJM=J$ which implies $M\in \limfunc{Sp}%
(n)$ hence $\Omega =S(B^{2n}(\sqrt{\hbar }))$ with $S=M^{-1/2}\in \limfunc{Sp%
}(n)$.
\end{proof}

Proposition \ref{Propinter} relating orthogonal projections and
intersections generalizes as follows to the case of symplectic polar duality:

\begin{proposition}
\label{Propintersymp}Let $\Omega \subset $ $\mathbb{R}^{2n}$ be a centrally
symmetric convex body and $F$ a linear subspace of $\mathbb{R}^{2n}$. We have%
\begin{equation}
(\Pi _{F}\Omega )^{\hbar ,\sigma }=\Omega ^{\hbar ,\sigma }\cap (JF)\text{ \
and }(\Omega \cap F)^{\hbar ,\sigma }=\Pi _{JF}(\Omega ^{\hbar ,\sigma }).
\label{orthogonal}
\end{equation}
\end{proposition}

\begin{proof}
Since $\Omega ^{\hbar }$ is symmetric we have $J(\Omega \cap F)^{\hbar
}=-J(\Omega \cap F)^{\hbar }$ hence the kernel of the projector $-J\Pi _{F}J$
is $F$ and its range is $J\ell $ so that $-J\Pi _{F}J=\Pi _{JF}$. This
proves the first equality (\ref{orthogonal}). We have, by definition, $%
(\Omega \cap F)^{\hbar ,\sigma }=J(\Omega \cap F)^{\hbar }$. In view of
formula (\ref{projdual}) we have $(\Omega \cap F)^{\hbar }=\Pi _{F}(\Omega
^{\hbar })$ and hence $(\Omega \cap F)^{\hbar ,\sigma }=J\Pi _{F}(\Omega
^{\hbar })$. Thus%
\begin{equation*}
(\Pi _{F}\Omega )^{\hbar ,\sigma }=-J\Pi _{F}(\Omega ^{\hbar })=(-J\Pi
_{F}J)(\Omega ^{\hbar ,\sigma });
\end{equation*}%
which is the second equality (\ref{orthogonal}). It immediately follows from
Proposition \ref{Propinter} noting that 
\begin{equation*}
(\Pi _{F}\Omega )^{\hbar ,\sigma }=J(\Omega ^{\hbar }\cap F)=\Omega ^{\hbar
,\sigma }\cap JF.
\end{equation*}
\end{proof}

We finally note that the Blaschke--Santal\'{o} inequality (\ref{sant1})
becomes in this context%
\begin{equation}
\limfunc{Vol}\nolimits_{2n}(\Omega )\limfunc{Vol}\nolimits_{2n}(\Omega
^{\hbar ,\sigma }).\leq (\limfunc{Vol}\nolimits_{2n}(B^{2n}(\sqrt{\hbar }%
))^{2}  \label{sant2}
\end{equation}%
with equality if and only the convex set $X$ is an ellipsoid. This follows
from (\ref{sant1}) noting that 
\begin{equation*}
\limfunc{Vol}\nolimits_{2n}(\Omega ^{\hbar ,\sigma })=\limfunc{Vol}%
\nolimits_{2n}(J\Omega ^{\hbar })=\limfunc{Vol}\nolimits_{2n}(\Omega ^{\hbar
}).
\end{equation*}

\subsection{Quantum Admissible Ellipsoids}

\subsubsection{Definition and a necessary and sufficient condition}

The following definition will be motivated below:

\begin{definition}
\label{defadm}Let $\Omega _{M}\subset \mathbb{R}^{2n}$ be the ellipsoid $%
\{z:Mz\cdot z\leq \hbar \}$ ($M=M^{t}>0$). We will say that $\Omega _{M}$ is
quantum admissible if it contains a quantum blob\ $Q_{S}=S(B^{2n}(\sqrt{%
\hbar }))$, $S\in \limfunc{Sp}(n)$.
\end{definition}

We are going to prove, using symplectic polarity, two simple but important
necessary and sufficient conditions for an ellipsoid to be quantum
admissible. We first recall the Williamson symplectic diagonalization result 
\cite{Williamson} (see \cite{Birk,HZ} for \textquotedblleft
modern\textquotedblright\ proofs). For every $M=M^{T}>0$, there exists $S\in 
\limfunc{Sp}(n)$ such that%
\begin{equation}
M=S^{T}DS\text{ \ , \ }D=%
\begin{pmatrix}
\Lambda ^{\sigma } & 0_{n\times n} \\ 
0_{n\times n} & \Lambda ^{\sigma }%
\end{pmatrix}
\label{Williamson}
\end{equation}%
where $\Lambda ^{\sigma }=\limfunc{diag}(\lambda _{1}^{\sigma },...,\lambda
_{n}^{\sigma })$ the $\lambda _{j}^{\sigma }$ being the symplectic
eigenvalues of $M$ (\textit{i.e.} the moduli of the eigenvalues of $JM\sim
M^{1/2}JM^{1/2}$). It is usual to rank the symplectic eigenvalues in
non-increasing order%
\begin{equation*}
\lambda _{\max }^{\sigma }=\lambda _{1}^{\sigma }\geq \lambda _{2}^{\sigma
}\geq \cdot \cdot \cdot \geq \lambda _{n}^{\sigma }=\lambda _{\min }^{\sigma
}.
\end{equation*}%
Note that the symplectic spectrum of $M^{-1}$ is $((\lambda _{1}^{\sigma
})^{-1},...,(\lambda _{n}^{\sigma })^{-1})$. It is usual to call the
factorization (\ref{Williamson}) the \textquotedblleft Williamson normal
form of $M$\textquotedblright .

\begin{remark}
The diagonalizing symplectic matrix $S$ in (\ref{Williamson}) is not unique;
see Son \textit{et al.} \cite{son} for a detailed analysis of the set of
diagonalizing symplectic matrices.
\end{remark}

Recall from Proposition \ref{propfix} that the equality $\Omega _{M}^{\hbar
,\sigma }=\Omega _{M}$ occurs if and only if $\Omega _{M}=S(B^{2n}(\sqrt{%
\hbar }))$ for some $S\in \limfunc{Sp}(n)$, \textit{i.e.} if and only if $%
\Omega $ is a \textquotedblleft quantum blob\textquotedblright . \ Below we
state and prove a general criterion for admissibility which we glorify it by
giving it the status of a theorem. Let us first introduce some preparatory
material:

\begin{itemize}
\item A two-dimensional subspace $F$ of $(\mathbb{R}^{2n},\sigma )$ is
called a symplectic plane if the restriction $\sigma |F$ of symplectic form $%
\sigma $ is non-degenerate; equivalently $F$ has a basis $\{e_{1},e_{2}\}$
such that $\sigma (e_{1},e_{2})=1$. In particular every plane $F_{j}$ of
conjugate coordinates $x_{j},p_{j}$ is symplectic; and for every symplectic
plane $F$ there exists $S_{j}\in \limfunc{Sp}(n)$ such that $F=S_{j}(F_{j})$.

\item We will use Gromov's symplectic non-squeezing theorem \cite{gr85}; it
says (in its simplest form) that no symplectomorphism $f\in \limfunc{Symp}%
(n) $ of $(\mathbb{R}^{2n},\sigma )$ can send a ball $B^{2n}(R)$ into a
cylinder $Z_{j}^{2n}(r):x_{j}^{2}+p_{j}^{2}\leq r^{2}$ if $r<R$ (we are
denoting by $\limfunc{Symp}(n)$ the group of all symplectomorphisms \cite%
{Leonid,HZ,zehnder} of $(\mathbb{R}^{2n},\sigma )$).
\end{itemize}

We will also need the following immediate consequence of Gromov's theorem:

\begin{lemma}
\label{Gromov}Let $F$ be a symplectic plane in $(\mathbb{R}^{2n},\sigma )$
and $f\in \limfunc{Symp}(n)$. The area of the orthogonal projection $\Pi
_{F} $ of \ $f(B^{2n}(z_{0},r))$ on $F$ satisfies 
\begin{equation}
\func{Area}(\Pi _{F}f(B^{2n}(z_{0},r))))\geq \pi r^{2}.  \label{gromov1}
\end{equation}
\end{lemma}

\begin{proof}
It is sufficient to suppose that $z_{0}=0$ since areas are
translation-invariant. Assume that $\Pi _{F_{j}}(S(B^{2n}(r)))=\pi R^{2}$
with $R<r$. Then $f(B^{2n}(r)$ must be contained in the cylinder $Z_{j}(R)$,
but this contradicts Gromov's non-squeezing theorem. (ii) Assume that $R<r$.
Then, by (\ref{gromov1}),%
\begin{equation*}
\pi r^{2}\leq \func{Area}(\Pi _{F}S(B^{2n}(R))))<\func{Area}(\Pi
_{F}S(B^{2n}(r)))).
\end{equation*}
\end{proof}

Let us now state and prove our theorem:

\begin{theorem}
\label{Thm2}The ellipsoid $\Omega _{M}$ is quantum admissible if and only if
the two following equivalent conditions are satisfied: (i) We have the
inclusion 
\begin{equation}
\Omega _{M}^{\hbar ,\sigma }\subset \Omega _{M}.  \label{in}
\end{equation}%
(ii) We have the inequality%
\begin{equation}
\func{Area}(\Omega _{M}^{\hbar ,\sigma }\cap F)\leq \pi \hbar
\label{areadualF}
\end{equation}%
for every symplectic plane $F$ in $(\mathbb{R}^{2n},\sigma )$.
\end{theorem}

\begin{proof}
(i) Suppose that $\Omega _{M}$ is quantum admissible; then there exists $%
S\in \limfunc{Sp}(n)$ such that $Q_{S}=S(B^{2n}(\sqrt{\hbar }))\subset
\Omega _{M}$. By the anti-monotonicity of symplectic polar duality this
implies that we have 
\begin{equation*}
\Omega _{M}^{\hbar ,\sigma }\subset Q_{S}^{\hbar ,\sigma }=Q_{S}\subset
\Omega _{M},
\end{equation*}%
which proves the necessity of the condition. Suppose conversely that $\Omega
_{M}^{\hbar ,\sigma }\subset \Omega _{M}$. We have 
\begin{equation}
\Omega _{M}^{\hbar ,\sigma }=\{z\in \mathbb{R}^{2n}:(-JM^{-1}J)z\cdot z\leq
\hbar \}  \label{ommom}
\end{equation}%
hence the inclusion $\Omega _{M}^{\hbar ,\sigma }\subset \Omega _{M}$
implies that $M\leq (-JM^{-1}J)$ ($\leq $ stands here for the L\"{o}wner
ordering of matrices). Performing a symplectic diagonalization (\ref%
{Williamson}) of $M$ and using the relations $JS^{-1}=S^{T}J$, $%
(S^{T})^{-1}J=JS$ this is equivalent to 
\begin{equation*}
M=S^{T}DS\leq S^{T}(-JD^{-1}J)S
\end{equation*}%
that is to $D\leq -JD^{-1}J$. This implies that we have $\Lambda ^{\sigma
}\leq (\Lambda ^{\sigma })^{-1}$ and hence $\lambda _{j}^{\sigma }\leq 1$
for $1\leq j\leq n$; thus $D\leq I$ and $M=S^{T}DS\leq S^{T}S$. The
inclusion $S(B^{2n}(\sqrt{\hbar }))\subset \Omega _{M}$ follows and we are
done. (ii) Suppose that $\Omega _{M}$ is admissible and let $\Pi _{F}$ be
the orthogonal projection in $\mathbb{R}^{2n}$ on $F$. By Proposition \ref%
{Propintersymp} we have%
\begin{equation*}
\Omega _{M}^{\hbar ,\sigma }\cap F=(\Pi _{JF}\Omega _{M})^{\hbar ,\sigma }.
\end{equation*}%
Since $\Omega _{M}$ is an ellipsoid the symplectic version (\ref{sant2}) of
the Blaschke--Santal\'{o} inequality becomes the equality%
\begin{equation}
\func{Area}(\Pi _{JF}\Omega _{M})^{\hbar ,\sigma }\func{Area}(\Pi
_{JF}\Omega _{M})=(\pi \hbar )^{2}  \label{BSarea}
\end{equation}%
that is%
\begin{equation*}
\func{Area}(\Omega _{M}^{\hbar ,\sigma }\cap F)\func{Area}(\Pi _{JF}\Omega
_{M})=(\pi \hbar )^{2}.
\end{equation*}%
The inequality (\ref{areadualF}) follows: since $\Omega _{M}$ is admissible,
it contains a quantum blob $(B^{2n}(\sqrt{\hbar }))$ hence $\func{Area}(\Pi
_{JF}\Omega _{M})\geq \pi \hbar $ in view of Lemma \ref{Gromov}. Assume
conversely that $\func{Area}(\Omega _{M}^{\hbar ,\sigma }\cap F)\leq \pi
\hbar $ for every symplectic plane $F$; by (\ref{BSarea}) we must then have $%
\func{Area}(\Pi _{JF}\Omega _{M})\geq \pi \hbar $ for every $F$. Let us show
that this implies that $\Omega _{M}$ must be admissible. Since admissibility
is preserved by symplectic conjugation we may assume, using a Williamson
diagonalization (\ref{Williamson}), that $M$ is of the diagonal type $%
\begin{pmatrix}
\Lambda ^{\sigma } & 0_{n\times n} \\ 
0_{n\times n} & \Lambda ^{\sigma }%
\end{pmatrix}%
$ where $\Lambda ^{\sigma }=\limfunc{diag}(\lambda _{1}^{\sigma
},...,\lambda _{n}^{\sigma })$ the $\lambda _{j}^{\sigma }$ being the
symplectic eigenvalues of $M.$ The ellipsoid $\Omega _{M}$ is thus given by 
\begin{equation*}
\lambda _{1}^{\sigma }(x_{1}^{2}+p_{1}^{2})+\cdot \cdot \cdot +\lambda
_{n}^{\sigma }(x_{n}^{2}+p_{n}^{2})\leq \hbar .
\end{equation*}%
Let us intersect $\Omega _{M}$ with the \ symplectic plane $F_{1}$ (the
plane of coordinates $x_{1},p_{1}$). It is the ellipse $x_{1}^{2}+p_{1}^{2}%
\leq \hbar /\lambda _{1}^{\sigma }$ which has area $\pi \hbar /\lambda
_{1}^{\sigma }$. Now, $\Omega _{M}$ is admissible if and only if $\lambda
_{1}^{\sigma }=\lambda _{\max }^{\sigma }\leq 1$ which is equivalent to the
condition $\func{Area}(\Omega _{M}\cap F)\geq \pi ,$ hat is to $\func{Area}%
(\Omega _{M}^{\hbar ,\sigma }\cap F)\leq \pi \hbar $ again in view of the
Blaschke--Santal\'{o} equality (\ref{BSarea}).
\end{proof}

Condition (\ref{areadualF}) is truly remarkable; it shows that given an
ellipsoid and its symplectic polar dual the datum of a sequence of
two-dimensional conditions suffices to decide whether the ellipsoid is
admissible or not. This \textquotedblleft tomographic\textquotedblright\
property is related to a condition using the Poincar\'{e} invariant given by
Narcowich \cite{Narcow} we will briefly discuss in our study of covariance
and information ellipsoids in Section \ref{secco}.

\subsection{Sub-Gaussian estimates for the Wigner function\label{secsubgauss}%
}

Recall \cite{Wigner} that the cross-Wigner function of a pair $(\psi ,\phi
)\in L^{2}(\mathbb{R}^{n})\times L^{2}(\mathbb{R}^{n})$ it is defined by 
\begin{equation}
W(\psi ,\phi )(z)=\left( \tfrac{1}{2\pi \hbar }\right) ^{n}\int_{\mathbb{R}%
^{n}}e^{-\frac{i}{\hbar }p\cdot y}\psi (x+\tfrac{1}{2}y)\overline{\phi (x-%
\tfrac{1}{2}y)}dy~.  \label{8}
\end{equation}%
$W(\psi ,\phi )$ is a continuous function satisfying the estimate 
\begin{equation}
|W(\psi ,\phi )(z)|\leq \left( \tfrac{2}{\pi \hbar }\right) ^{n}||\psi
||_{L^{2}}||\phi ||_{L^{2}}.  \label{bound}
\end{equation}%
When $\psi =\phi $ the function $W(\psi ,\phi )=W\psi W(\psi ,\phi )$ is the
usual Wigner function 
\begin{equation}
W\psi (z)=\left( \tfrac{1}{2\pi \hbar }\right) ^{n}\int_{\mathbb{R}^{n}}e^{-%
\frac{i}{\hbar }p\cdot y}\psi (x+\tfrac{1}{2}y)\overline{\psi (x-\tfrac{1}{2}%
y)}dy.  \label{Wig}
\end{equation}%
The Wigner functions of general Gaussian functions is well-known \cite%
{Birk,Wigner}; if 
\begin{equation}
\psi _{A,B}(x)=\left( \tfrac{1}{\pi \hbar }\right) ^{n/4}(\det A)^{1/4}e^{-%
\tfrac{1}{2\hbar }A+iB)x\cdot x}  \label{psix1}
\end{equation}%
where $A,B\in \limfunc{Sym}(n,\mathbb{R})$, $A>0$. Then \cite{Birk,Wigner} 
\begin{equation}
W\psi _{A,B}(z)=\left( \tfrac{1}{\pi \hbar }\right) ^{n}e^{-\frac{1}{h}%
Gz\cdot z}  \label{psix2}
\end{equation}%
where $G\in \limfunc{Sp}(n)$: 
\begin{equation}
G=S^{T}S\in \limfunc{Sp}(n)\text{ \ , \ }S=%
\begin{pmatrix}
A^{1/2} & 0_{n\times n} \\ 
A^{-1/2}B & A^{-1/2}%
\end{pmatrix}%
.  \label{psix3}
\end{equation}%
Explicitly%
\begin{equation}
G=%
\begin{pmatrix}
A+BA^{-1}BY & BA^{-1} \\ 
A^{-1}B & A^{-1}%
\end{pmatrix}%
.  \label{gsymp}
\end{equation}%
Notice that, conversely, if $\psi \in L^{2}(\mathbb{R}^{n})$ is such that $%
W(z)=(\pi \hbar )^{-n}e^{-\frac{1}{h}Mz\cdot z}$ for some $M=M^{T}\in 
\limfunc{Sp}(n)$, $M>0$, then $\psi =e^{i\chi }\psi _{A,B}$ ($\chi \in 
\mathbb{R}$) where $X$ and $Y$ are determined by (\ref{gsymp}) and $\chi |=1$%
.

Sub-Gaussian estimates for the Wigner function refer to bounds on the
magnitude of the Wigner function that ensure that it does not fluctuate too
much. A function is sub-Gaussian if its tails decay faster than any Gaussian
distribution. Using Proposition \ref{prophardy1} one proves \cite%
{Wigner,golulett} that:

\begin{corollary}
\label{corwig}Let $\psi \in L^{2}(\mathbb{R}^{n})$, $\psi \neq 0$, and
assume that there exists $C>0$ such that $W\psi (z)\leq Ce^{-\frac{1}{\hbar }%
Mz\cdot z}$ where $M=M^{T}>0$. Then the symplectic eigenvalues $\lambda
_{1}^{\sigma }\geq \lambda _{2}^{\sigma }\geq \cdot \cdot \cdot \geq \lambda
_{n}^{\sigma }$ of $M$ are all $\leq 1$. When $\lambda _{1}^{\sigma
}=\lambda _{2}^{\sigma }=\cdot \cdot \cdot =\lambda _{n}^{\sigma }=1$ then
the function $\psi $ is a generalized Gaussian (\ref{psix1}).
\end{corollary}

It follows from this result that the Wigner function $W\psi $ can never have
compact support: assume that there exists $R>0$ such that $W\psi (z)=0$ for $%
|z|>R$. Then, for every $a>0$ there exists a constant $C(a)>0$ such that $%
W\psi (z)\leq C(a)e^{-\frac{a}{\hbar }|z|^{2}}$ for all $z\in \mathbb{R}%
^{2n} $. Choosing $a$ large enough this contradicts the statement in
Corollary \ref{corwig} because as soon as $a>1$ the symplectic eigenvalues
of $M=aI_{n\times n}$ are all equal to $a.$

Corollary \ref{corwig} can be elegantly reformulated in terms of symplectic
polar duality:

\begin{proposition}
The Wigner function $W\psi $ of $\psi \in L^{2}(\mathbb{R}^{n})$ satisfies a
sub-Gaussian estimate $W\psi (z)\leq Ce^{-\frac{1}{\hbar }Mz\cdot z}$ if and
only if the ellipsoid $\Omega _{M}=\{z:Mz\cdot z\leq \hbar \}$ is
admissible: $\Omega _{M}^{\hbar ,\sigma }\subset \Omega _{M}$, that is, $%
\Omega _{M}$ contains a quantum blob $S(B^{2n}(\hbar ))$, $S\in \limfunc{Sp}%
(n)$.
\end{proposition}

\begin{proof}
In view of Williamson's diagonalization theorem \ref{Williamson} there
exists $S\in \limfunc{Sp}(n)$ such that $M=S^{T}DS$ where $D=%
\begin{pmatrix}
\Lambda ^{\sigma } & 0_{n\times n} \\ 
0_{n\times n} & \Lambda ^{\sigma }%
\end{pmatrix}%
$ hence $S(\Omega _{M})=\Omega _{D}$. Since $S(\Omega ^{\hbar ,\sigma
})=(S(\Omega ))^{\hbar ,\sigma }$ (Proposition \ref{Prop1}) \ it is
sufficient to prove the result for $M=D$. Since $\Omega _{D}^{\hbar ,\sigma
}=\Omega _{-JD^{-1}J}$ we thus have to show that the sub-Gaussian estimate
is satisfied if and only if $\Omega _{-JD^{-1}J}\subset \Omega _{D}$. This
is equivalent to $D\leq -JD^{-1}J$ (in the L\"{o}wner ordering) that is to $%
\Lambda ^{\sigma }\leq (\Lambda ^{\sigma })^{-1}$ which is possible if and
only the symplectic eigenvalues $\lambda _{j}^{\sigma }$ are all $\leq 1$.
When the $\lambda _{j}^{\sigma }$ are all equal to one we have $%
D=I_{2n\times 2n}$., that is $M=S^{T}S$ so that $W\psi (z)=Ce^{-\frac{1}{%
\hbar }Mz\cdot z}$ for some constant $C$.
\end{proof}

\section{Modulation Spaces and Covariance Matrices}

\subsection{The modulation spaces $M_{s}^{1}$}

\subsubsection{Definition using the Wigner function}

Let $L_{s}^{1}(\mathbb{R}^{2n})$ is the weighted $L^{1}$-space defined by%
\begin{equation}
L_{s}^{1}(\mathbb{R}^{2n})=\{\rho :\mathbb{R}^{2n}\longrightarrow \mathbb{C}%
:\langle z\rangle ^{s}\rho \in L^{1}(\mathbb{R}^{2n})\}  \label{6}
\end{equation}%
where $\langle z\rangle =(1+|z|^{2})^{1/2}$. In the definition below we are
following our presentation of modulation spaces given in \cite{Birkbis}; see 
\cite{fe06,Gro,Jakobsen} for a definitions used in time-frequency analysis
(they are based on the short-time Fourier transform (Gabor transform)). We
denote by $\mathcal{S}(\mathbb{R}^{n})$ the Schwartz space of test functions
decreasing rapidly to zero at infinity, together with their derivatives.

\begin{definition}
The modulation space $M_{s}^{1}(\mathbb{R}^{n})$ ($s\in \mathbb{R}$)
consists of all $\psi \in L^{2}(\mathbb{R}^{n})$ such that $W(\psi ,\phi
)\in L_{s}^{1}\mathbb{R}^{2n})$ for every $\phi \in \mathcal{S}(\mathbb{R}%
^{n})$. When $s=0$ the space $M_{0}^{1}(\mathbb{R}^{n})=S_{0}(\mathbb{R}%
^{n}) $ is called the \textit{Feichtinger algebra}.
\end{definition}

It turns out that it suffices to check that condition $W(\psi ,\phi )\in
L_{s}^{1}\mathbb{R}^{2n})$ holds for \textit{one} function $\phi \neq 0$
(hereafter called \textquotedblleft window\textquotedblright\ for it then it
holds for \textit{all }$\phi \in \mathcal{S}(\mathbb{R}^{n})$.The mappings $%
\psi \longmapsto ||\psi ||_{\phi ,M_{s}^{1}}$ defined by 
\begin{equation*}
||\psi ||_{\phi ,M_{s}^{1}}=||W(\psi ,\phi )||_{L_{s}^{1}}=\int_{\mathbb{R}%
^{2n}}|W(\psi ,\phi )(z)|\left\langle z\right\rangle ^{s}dz
\end{equation*}%
form a family of equivalent norms, and the topology on $M_{s}^{1}(\mathbb{R}%
^{n})$ thus defined makes it into a Banach space. We have the chain of
inclusions 
\begin{equation*}
\mathcal{S}(\mathbb{R}^{n})\subset M_{s}^{1}(\mathbb{R}^{n})\subset
L_{s}^{1}(\mathbb{R}^{n})\cap F(L_{s}^{1}(\mathbb{R}^{n}))
\end{equation*}%
where $F$ is the Fourier transform; it follows by Riemann--Lebesgue that in
particular 
\begin{equation*}
M_{s}^{1}(\mathbb{R}^{n})\subset Ls^{1}(\mathbb{R}^{n})\cap C^{0}(\mathbb{R}%
^{n}).
\end{equation*}%
Observe that 
\begin{equation*}
M_{s}^{1}(\mathbb{R}^{n})\subset M_{s^{\prime }}^{1}(\mathbb{R}%
^{n})\Longleftrightarrow s\geq s^{\prime }
\end{equation*}%
and one proves \cite{Gro} that 
\begin{equation*}
\tbigcap_{s\geq 0}M_{s}^{1}(\mathbb{R}^{n})=\mathcal{S}(\mathbb{R}^{n}).
\end{equation*}

An essential property of the spaces $M_{s}^{1}(\mathbb{R}^{n})$ is their
metaplectic invariance for $s\geq 0$. Recall that the metaplectic group $%
\limfunc{Mp}(n)$ is the unitary representation in $L^{2}(\mathbb{R}^{n})$ of
the double cover $\limfunc{Sp}_{2}(n)$ of the symplectic group $\limfunc{Sp}%
(n)$ (see for instance \cite{Birk} for a detailed study of the metaplectic
representation). We will denote $\pi ^{\limfunc{Mp}}$ the covering
projection $\limfunc{Mp}(n)\longrightarrow \limfunc{Sp}(n)$; it is a
two-to-one epimorphism.

\begin{proposition}
\label{PropMp1}The modulation spaces $M_{s}^{1}(\mathbb{R}^{n})$, $s\geq 0$,
are invariant under the action of $\limfunc{Mp}(n)$: if $\psi \in M_{s}^{1}(%
\mathbb{R}^{n})$ and $\widehat{S}\in \limfunc{Mp}(n)$ then $\widehat{S}\psi
\in M_{s}^{1}(\mathbb{R}^{n})$.
\end{proposition}

\begin{proof}
By definition we have $\psi \in M_{s}^{1}(\mathbb{R}^{n})$ if and only if $%
W(\psi ,\phi )\in L_{s}^{1}\mathbb{R}^{2n})$ for every $\phi \in \mathcal{S}(%
\mathbb{R}^{n})$; similarly $\widehat{S}\psi \in M_{s}^{1}(\mathbb{R}^{n})$
if and only $W(\widehat{S}\psi ,\phi )\in L_{s}^{1}\mathbb{R}^{2n})$ for
every $\phi $. Now, in view of the symplectic covariance of the cross-Wigner
function \cite{Birk,Birkbis} 
\begin{equation}
W(\widehat{S}\psi ,\phi )=W(\widehat{S}\psi ,\widehat{S}(\widehat{S}%
^{-1}\phi ))=W(\psi ,\widehat{S}^{-1}\phi )\circ S^{-1}  \label{syw}
\end{equation}%
where $S=\pi ^{\limfunc{Mp}}(\widehat{S})$. Since $\widehat{S}^{-1}\phi \in 
\mathcal{S}(\mathbb{R}^{n})$ there remains to prove that $W(\psi ,\widehat{S}%
^{-1}\phi )\circ S^{-1}\in L_{s}^{1}\mathbb{R}^{2n})$. Set $\phi ^{\prime }=%
\widehat{S}^{-1}\phi $; we have $\det S=1$ hence 
\begin{equation*}
\int_{\mathbb{R}^{2n}}|W(\psi ,\phi ^{\prime })(S^{-1}z)|\left\langle
z\right\rangle ^{s}dz=\int_{\mathbb{R}^{2n}}|W(\psi ,\phi ^{\prime
})(z)|\left\langle Sz\right\rangle ^{s}dz.
\end{equation*}%
We have $\left\langle Sz\right\rangle ^{s}\leq C_{S}\left\langle
z\right\rangle ^{s}$ for some constant $C_{S}>0$, thus%
\begin{equation*}
\int_{\mathbb{R}^{2n}}|W(\psi ,\phi ^{\prime })(S^{-1}z)|\left\langle
z\right\rangle ^{s}dz\leq C_{S}\int_{\mathbb{R}^{2n}}|W(\psi ,\phi ^{\prime
})(z)|\left\langle z\right\rangle ^{s}dz<\infty
\end{equation*}%
and we are done.
\end{proof}

\subsection{Density operators and covariance matrices\label{secdensity}}

Quantum mechanics is inherently probabilistic, and the behavior of quantum
systems is described by mathematical entities known as quantum states. These
states are represented by density operators (also called density matrices by
physicists). An element $\widehat{\rho }$ of the algebra $\mathcal{L}^{1}(%
\mathbb{R}^{n})$ of trace class operators on $L^{2}(\mathbb{R}^{n})$ is
called a density operator if $\widehat{\rho }\geq 0$ and has trace $\limfunc{%
Tr}\widehat{\rho }=1$; the passiveness implies in particular that $\widehat{%
\rho }$ is a self-adjoint compact operator \ so that there exists a sequence 
$(\alpha _{j})$ with $\alpha _{j}\geq 0$, $\sum_{j}\alpha _{j}=1$ and an
orthonormal stem $(\psi _{j})$, $\psi _{j}\in L^{2}(\mathbb{R}^{n})$ such
that we have we have the spectral decomposition 
\begin{equation}
\widehat{\rho }=\sum_{j}\alpha _{j}\widehat{\rho }_{j}\text{ }\ ,\text{ \ }%
\alpha _{j}\geq 0\text{ \ },\text{ }\sum_{j}\alpha _{j}=1  \label{ro}
\end{equation}%
where $\widehat{\rho }_{j}$ is the rank one orthogonal projection in $L^{2}(%
\mathbb{R}^{n})$ on the ray $\mathbb{C}\psi _{j}$. It follows that that the
Weyl symbol of the operator $\widehat{\rho }$ is \cite{Birkbis}%
\begin{equation}
(2\pi \hbar )^{n}\rho =\sum_{j}\alpha _{j}W\psi _{j}  \label{rhowig}
\end{equation}%
(the function $\rho $ is often called the \textit{Wigner distribution} of ~$%
\widehat{\rho }$ in the physical literature). In \cite{CM,QHA} we introduced
the notion of \textit{Feichtinger state}:

\begin{definition}
A density operator $\widehat{\rho }\in \mathcal{L}^{1}(\mathbb{R}^{n})$ is
called a a \textquotedblleft Feichtinger state\textquotedblright\ if each $%
\psi _{j}\in M_{s}^{1}(\mathbb{R}^{n})$ for some $s\geq 0$.
\end{definition}

Feichtinger states satisfy the marginal properties: we have $\rho \in L^{1}(%
\mathbb{R}^{2n})$ \ and%
\begin{equation}
\int_{\mathbb{R}^{n}}\rho (x,p)dp=\sum_{j}\alpha _{j}|\psi _{j}(x)|^{2}\text{
},\text{ }\int_{\mathbb{R}^{n}}\rho (x,p)dx=\sum_{j}\alpha _{j}|F\psi
_{j}(x)|^{2}.  \label{marg2}
\end{equation}

The main interest of the notion of Feichtinger state in our context is that
they also allow to define rigorously the covariance matrix $\Sigma _{\mathrm{%
cov}}$ of a density operator. The latter\ is defined --if it exists!-- as
being the symmetric $2n\times 2\times n$ matrix%
\begin{equation}
\Sigma _{\mathrm{cov}}=\int\nolimits_{\mathbb{R}^{2n}}(z-\overline{z})(z-%
\overline{z})^{T}\rho (z)dz  \label{vectorform}
\end{equation}%
where $\overline{z}=\int\nolimits_{\mathbb{R}^{2n}}z\rho (z)dz$ is the
average (or mean value) vector. It is convenient to write the covariance
matrix in $n\times n$ block-matrix form as 
\begin{equation}
\Sigma _{\mathrm{cov}}=%
\begin{pmatrix}
\Sigma _{XX} & \Sigma _{XP} \\ 
\Sigma _{PX} & \Sigma _{PP}%
\end{pmatrix}%
\text{ \ },\text{ \ }\Sigma _{PX}=\Sigma _{XP}^{T}  \label{covmat}
\end{equation}%
where $\Sigma _{XX}=(\sigma _{x_{j}x_{k}})_{1\leq j,k,\leq n}$, $\Sigma
_{PP}=(\sigma _{p_{j}p_{k}})_{1\leq j,k,\leq n}$, and $\Sigma _{XP}=(\sigma
_{x_{j}p_{k}})_{1\leq j,k,\leq n}$ with%
\begin{equation}
\sigma _{x_{j}x_{k}}=\int_{\mathbb{R}^{2n}}x_{j}x_{k}\rho (z)dz  \label{3}
\end{equation}%
and so on. We have the following conjugation result:

\begin{proposition}
The covariance matrix $\Sigma _{\mathrm{cov}}$ of a Feichtinger state $%
\widehat{\rho }$ \ with $s\geq 2$ is well-defined. If $\widehat{S}\in 
\limfunc{Mp}(n)$ then conjugate state $\widehat{S}\widehat{\rho }\widehat{S}%
^{-1}$ is also a Feichtinger state with $s\geq 2$ and the covariance matrix
of $\widehat{S}\widehat{\rho }\widehat{S}^{-1}$ is $S\Sigma _{\mathrm{cov}%
}S^{T}$.
\end{proposition}

\begin{proof}
Since $s\geq 2$ we have 
\begin{equation}
\int\nolimits_{\mathbb{R}^{2n}}|\rho (z)|(1+|z|^{2})dz<\infty ~.
\label{cond25}
\end{equation}%
It is no restriction to assume $\overline{z}=0$; setting $z_{\alpha
}=x_{\alpha }$ if $1\leq \alpha \leq n$ and $z_{\alpha }=p_{\alpha }$ if $%
n+1\leq \alpha \leq 2n$ we have $\Sigma =(\sigma _{\alpha \beta })_{1\leq
\alpha ,\beta \leq 2n}$ where the integrals 
\begin{equation*}
\sigma _{\alpha \beta }=\int\nolimits_{\mathbb{R}^{2n}}z_{\alpha }z_{\beta
}\rho (z)dz
\end{equation*}%
are absolutely convergent in view of the trivial inequalities $|z_{\alpha
}z_{\beta }|\leq 1+|z|^{2}$. Let $\widehat{S}\in \limfunc{Mp}(n)$ ; it
follows from the standard properties of Weyl pseudodifferential calculus 
\cite{Birk,Birkbis,sh87} that the Weyl symbol of the conjugate $\widehat{S}%
\widehat{\rho }\widehat{S}^{-1}$ is$(2\pi \hbar )^{n}\rho \circ S^{-1}$. In
view of Proposition \ref{PropMp1} $\rho \in M_{s}^{1}(\mathbb{R}^{n})$
implies that $\rho \circ S\in M_{s}^{1}(\mathbb{R}^{n})$ and the result
follows from (\ref{vectorform}) by a simple calculation.
\end{proof}

\subsubsection{A necessary condition for positivity\label{secpos}}

Assume from now on that $\widehat{\rho }$ is a Feichtinger state with $s\geq
2$, guaranteeing the existence of the covariance matrix $\Sigma _{\mathrm{cov%
}}$. It i a well-known property in harmonic analysis that the positivity
condition $\widehat{\rho }$ $\geq 0$ implies that $\Sigma $ must satisfy the
algebraic condition%
\begin{equation}
\Sigma +\frac{i\hbar }{2}J\text{ \emph{is positive semidefinite}}
\label{quantum}
\end{equation}%
which we write for short $\Sigma +\frac{i\hbar }{2}J\geq 0$. (Note that the
matrix $\Sigma +\frac{i\hbar }{2}J$ is selfadjoint since $\Sigma $ is
symmetric and $J^{\ast }=-J$.) This condition implies in particular that $%
\Sigma $ is positive definite \cite{Narcow,Birkbis} and hence invertible.
This (highly nontrivial) result can be proven by various methods, one can
for instance \ use the notion of $\eta $-positivity due to Kastler \cite%
{Kastler} together with a variant of Bochner's theorem on the Fourier
transform of a probability measure; for a simpler approach using methods
from harmonic analysis see our recent paper \cite{cogoni2} with Cordero and
Nicola. The condition (\ref{quantum}) actually first appeared as a compact
formulation of the uncertainty principle in Arvind \textit{et al. \cite%
{Arvind}. }We have rigorously shown in \cite{Birk,go09} that (\ref{quantum})
is equivalent to the Robertson--Schr\"{o}dinger inequalities 
\begin{equation}
\sigma _{x_{j}x_{j}}\sigma _{p_{j}p_{j}}\geq \sigma _{x_{j}p_{j}}^{2}+\tfrac{%
1}{4}\hbar ^{2}\text{ \ \textit{for} }1\leq j\leq n  \label{RS}
\end{equation}%
which form the textbook statement of the complete uncertainty principle of
quantum mechanics \cite{Messiah}. While condition (\ref{quantum}) is
generally only a necessary condition for the positivity of a trace class
operator, it is also sufficient for operators with Gaussian Weyl symbol $%
(2\pi \hbar )^{n}\rho $ where 
\begin{equation}
\rho (z)=\frac{1}{(2\pi )^{n}\sqrt{\det \Sigma }}e^{-\frac{1}{2}\Sigma
^{-1}(z-z_{0})\cdot (z-z_{0})}  \label{Gaussian}
\end{equation}%
($\Sigma =\Sigma ^{T}>0$ playing the role of a covariance matrix $\Sigma _{%
\mathrm{cov}}$; see for instance \cite{Birk,Birkbis,cogoni2,Rodino}).

\section{Symplectic Capacities and Polar Duality}

\subsection{The notion of symplectic capacity}

Symplectic capacities were defined by Ekeland and Hofer \cite{ekhof1,ekhof2}
(see \cite{cielibak,HZ} for review of that notion). They are closely related
to Gromov's symplectic non-squeezing theorem shortly discussed above; in
fact the existence of a single symplectic capacity is equivalent to Gromov's
theorem

Gromov's theorem ensures us of the existence of \textit{symplectic capacities%
}. A (normalized) symplectic capacity on $(\mathbb{R}^{2n},\sigma )$
associates to every subset $\Omega \subset \mathbb{R}^{2n}$ a number $%
c(\Omega )\in \mathbb{[}0,+\infty \mathbb{]}$ such that the following
properties hold (\cite{HZ}, see \cite{golu09} for a review):

\begin{description}
\item[SC1] \textit{Monotonicity}: If $\Omega\subset\Omega^{\prime}$ then $%
c(\Omega)\leq c(\Omega^{\prime})$;

\item[SC2] \textit{Conformality}: For every $\lambda\in\mathbb{R}$ we have $%
c(\lambda\Omega)=\lambda^{2}c(\Omega)$;

\item[SC3] \textit{Symplectic invariance}: $c(f(\Omega))=c(\Omega)$ for
every $f\in\limfunc{Symp}(n)$;

\item[SC4] \textit{Normalization}: For $1\leq j\leq n$, 
\begin{equation}
c(B^{2n}(r))=\pi r^{2}=c(Z_{j}^{2n}(r))  \label{cbz}
\end{equation}%
where $Z_{j}^{2n}(r)$ is the cylinder with radius $r$ based on the $%
x_{j},p_{j}$ plane.
\end{description}

It follows that the symplectic capacity of a quantum blob $%
Q_{S}(z_{0})=S(B^{2n}(\sqrt{\hbar }))$ is $c(Q_{S}(z_{0}))=\pi \hbar $.

There exist symplectic capacities, $c_{\max }$ and $c_{\min }$, such that $%
c_{\min }\leq c\leq c_{\max }$ for every symplectic capacity$c$, they are
defined by 
\begin{eqnarray}
c_{\max }(\Omega ) &=&\inf_{f\in \limfunc{Symp}(n)}\{\pi r^{2}:f(\Omega
)\subset Z_{j}^{2n}(r)\}  \label{cmax} \\
c_{\min }(\Omega ) &=&\sup_{f\in \limfunc{Symp}(n)}\{\pi
r^{2}:f(B^{2n}(r))\subset \Omega \}.  \label{cmin}
\end{eqnarray}%
That $c_{\min }$ and $c_{\min }$ indeed are symplectic capacities follows
from the axioms (SC1)---(SC4). Note that the conformality and normalization
properties (SC2) and (SC4) show that for $n>1$ symplectic capacities are not
related to volume; they have the dimension of an area \cite{Mellipsoid}, or
equivalently, that of an action. For instance, the Hofer--Zehnder capacity 
\cite{HZ} is characterized by the property that when $\Omega $ is a compact
convex set in $(\mathbb{R}^{2n},\sigma )$ with smooth boundary $\partial
\Omega $ then 
\begin{equation}
c_{\mathrm{HZ}\ }(\Omega )=\int_{\gamma _{\min }}pdx  \label{HZ}
\end{equation}%
where $pdx=p_{1}dx_{1}+\cdot \cdot \cdot +p_{n}dx_{n}$ and $\gamma _{\min }$
is the shortest positively oriented Hamiltonian periodic orbit carried by $%
\partial \Omega $.

One also has the weaker notion of linear (or affine) symplectic capacity,
obtained by replacing condition (SC3) with

\begin{description}
\item[SC3lin] \textit{Linear} \textit{symplectic invariance}: $c(S(\Omega
))=c(\Omega)$ for every $S\in\limfunc{Sp}(n)$ and $c(\Omega +z))=c(\Omega)$
for every $z\in\mathbb{R}^{2n}$.
\end{description}

The corresponding minimal and maximal linear symplectic capacities $c_{\min
}^{\mathrm{lin}}$ and $c_{\max }^{\mathrm{lin}}$ are then given by 
\begin{align}
c_{\min }^{\mathrm{lin}}(\Omega )& =\sup_{S\in \limfunc{Sp}(n)}\{\pi
R^{2}:S(B^{2n}(z,R)),z\in \mathbb{R}^{2n}\}  \label{clin1} \\
c_{\max }^{\mathrm{lin}}(\Omega )& =\inf_{f\in \limfunc{Sp}(n)}\{\pi
r^{2}:S(\Omega )\subset Z_{j}^{2n}(z,r),z\in \mathbb{R}^{2n}\}.
\label{clin2}
\end{align}

\subsubsection{The case of ellipsoids}

It turns out that all symplectic capacities (linear as well as non-linear)
agree on ellipsoids. Assume that 
\begin{equation*}
\Omega _{M}=\{z\in \mathbb{R}^{2n}:Mz\cdot z\leq r^{2}\}
\end{equation*}%
where $M\in \limfunc{Sym}^{+}(2n,\mathbb{R})$, and let $\lambda _{\max
}^{\sigma }=\lambda _{1}^{\sigma }\geq \lambda _{2}^{\sigma }\geq \cdot
\cdot \cdot \geq \lambda _{n}^{\sigma }$ be the symplectic eigenvalues of $M$%
. If in particular $\Omega =\Omega _{M}:Mz\cdot z\leq \hbar $ then 
\begin{equation}
c(\Omega _{M})=\pi \hbar /\lambda _{\max }^{\sigma }.  \label{com}
\end{equation}%
We have in particular 
\begin{equation*}
c_{\mathrm{HZ}\ }(\Omega _{M})=\int_{\gamma _{\min }}pdx=\pi \hbar /\lambda
_{\max }^{\sigma }.
\end{equation*}%
This is easy to verify using the calculations \ in the proof of Theorem \ref%
{Thm2}(ii): reducing $\Omega _{M}$ to Williamson normal form 
\begin{equation*}
\lambda _{1}^{\sigma }(x_{1}^{2}+p_{1}^{2})+\cdot \cdot \cdot +\lambda
_{n}^{\sigma }(x_{n}^{2}+p_{n}^{2})\leq \hbar
\end{equation*}%
the shortest Hamiltonian orbit carried by $\partial \Omega _{M}$ is given by
Hamilton's equations for the Hamiltonian function $H_{1}(x_{1},p_{1})=%
\lambda _{1}^{\sigma }(x_{1}^{2}+p_{1}^{2})$ with the condition $%
H_{1}(x_{1},p_{1})=\hbar $. One verifies that this periodic solution verifies%
\begin{equation*}
\int_{\gamma _{\min }}pdx=\pi \hbar /\lambda _{1}^{\sigma }=\pi \hbar
/\lambda _{\max }^{\sigma }=c(\Omega _{M}).
\end{equation*}

For the symplectic polar dual ellipsoid we have the following result, which
yields a Blaschke--Santal\'{o} type inequality for symplectic capacities of
ellipsoids:

\begin{proposition}
Let $\Omega _{M}$ be as above and let $\Omega _{M}^{\hbar ,\sigma }$ be its
symplectic polar dual. (i) We have%
\begin{equation}
c(\Omega _{M})=\pi \hbar /\lambda _{\max }^{\sigma }\text{ \ \ and \ }%
c(\Omega _{M}^{\hbar ,\sigma })=\pi \hbar \lambda _{\min }^{\sigma }
\label{capellipse}
\end{equation}%
where $\lambda _{\max }^{\sigma }$ (resp. $\lambda _{\min }^{\sigma }$) is
the largest (rep. smallest)symplectic eigenvalue of $M$. (ii)In particular 
\begin{equation}
c(\Omega _{M})c(\Omega _{M}^{\hbar ,\sigma })\leq (\pi \hbar )^{2}
\label{BSellipse}
\end{equation}%
with equality if and only if $\Omega _{M}=\lambda B^{2n}(\sqrt{\hbar })$ for
some $\lambda >0$.
\end{proposition}

\begin{proof}
(i) The formula $c(\Omega _{M})=\pi r^{2}/\lambda _{\max }^{\sigma }$ \ is
easily proven using a symplectic diagonalization (\ref{Williamson}) of $M$
which reduces the problem to the case of an ellipsoid with axes contained in
the conjugate $x_{j},p_{j}$ planes (see \cite{HZ,Birk} for details). The
symplectic polar dual of $\Omega _{M}$ is 
\begin{equation*}
\Omega _{M}^{\hbar ,\sigma }=\{z\in \mathbb{R}^{2n}:-JM^{-1}Jz\cdot z\leq
\hbar \}
\end{equation*}%
(formula (\ref{ohasig}) in Lemma \ref{lemmaelldual}). Set $N=-JM^{-1}J$: we
have $JN=M^{-1}J$ hence the eigenvalues of $JN$ are those of $%
M^{-1/2}JM^{-1/2}$ so the symplectic eigenvalues of N are the inverses of
those of $M$; the second formula (\ref{capellipse}) follows. (ii) Formula (%
\ref{BSellipse}) is obvious; that we have equality if and only if all $%
\Omega _{M}=\lambda B^{2n}(\sqrt{\hbar })$ follows from the fact that if the
symplectic eigenvalues of $M$ are all equal then by Williamson's theorem $M$
is a scalar multiple of a matrix $S^{T}S$.
\end{proof}

\subsection{Symplectic polarity and covariance ellipsoids\label{secco}}

\subsubsection{Covariance and information ellipsoids}

We are going to express condition (\ref{quantum}) in a simple geometric way
using symplectic polarity. By definition, the covariance ellipsoid of a
Feichtinger state $\widehat{\rho }$\ is 
\begin{equation}
\Omega _{\mathrm{cov}}=\{z\in \mathbb{R}^{2n}:\tfrac{1}{2}\Sigma _{\mathrm{%
cov}}^{-1}z\cdot z\leq 1\}  \label{defcovma}
\end{equation}%
where $\Sigma _{\mathrm{cov}}$ is the covariance matrix of $\widehat{\rho }$%
. The symplectic polar dual of $\Omega _{\mathrm{cov}}$ is the ellipsoid%
\begin{equation*}
\Omega _{\mathrm{cov}}^{\hbar ,\sigma }=\{z\in \mathbb{R}^{2n}:-\tfrac{1}{2}%
J\Sigma Jz\cdot z\leq 1\}.
\end{equation*}%
By definition the associated information (or precision) ellipsoid is 
\begin{equation*}
\Omega _{\mathrm{\inf o}}=\{z\in \mathbb{R}^{2n}:\tfrac{1}{2}\Sigma _{%
\mathrm{cov}}z\cdot z\leq 1\}
\end{equation*}%
by the equality $\Omega _{\mathrm{cov}}^{\hbar ,\sigma }=J(\Omega _{\mathrm{%
\inf o}})$, that is, $\Omega _{\mathrm{\inf o}}=J(\Omega _{\mathrm{cov}%
}^{\hbar })$ where $\Omega _{\mathrm{cov}}^{\hbar }$ is the ordinary polar
dual of $\Omega _{\mathrm{cov}}$. It turns out that $\Omega _{\mathrm{cov}}$
and $\Omega _{\mathrm{\inf o}}$ are Legendre duals of each other \cite%
{Narcow}. Consider in fact the quadratic forms $w(z)=\frac{1}{2}\Sigma _{%
\mathrm{cov}}^{-1}z\cdot z$ and $w^{\sigma }(z^{\prime })=\frac{1}{2}\Sigma
_{\mathrm{cov}}z^{\prime }\cdot z^{\prime }$. The Legendre transform of $%
w(z) $ is defined by 
\begin{equation*}
w^{\prime }(z^{\prime })=z\cdot z^{\prime }-w(z)
\end{equation*}
where $z$ is expressed in terms of $z^{\prime }$ by solving the equation $%
z^{\prime }=\partial _{z}w(z)=\Sigma _{\mathrm{cov}}^{-1}z$ hence $w^{\prime
}(z^{\prime })=w^{\sigma }(z^{\prime })$.

We are going to prove that the covariance ellipsoid of a Feichtinger state
always is quantum admissible:

\begin{proposition}
\label{Thm4}The covariance ellipsoid $\Omega _{\mathrm{cov}}$ of a
Feichtinger state $\widehat{\rho }$ satisfies the condition $\Omega _{%
\mathrm{cov}}^{\hbar ,\sigma }\subset \Omega _{\mathrm{cov}}$ and is hence
quantum admissible..
\end{proposition}

\begin{proof}
Setting set $M=\frac{\hbar }{2}\Sigma ^{-1}$ the covariance ellipsoid is
then defined by $\Omega _{M}:Mz\cdot z\leq \hbar $ and we have to prove that
the positivity condition $\Sigma +\frac{i\hbar }{2}J\geq 0$ is equivalent to
the condition $\Omega _{M}^{\hbar ,\sigma }\subset \Omega _{M}$. In view of
Propositions \ref{Thm2} it is sufficient to prove that this conditions hold
if and only if $\Omega _{M}$ contains a quantum blob $Q_{S}=S(B^{2n}(\sqrt{%
\hbar }))$, $S\in \limfunc{Sp}(n)$. Performing a symplectic diagonalization $%
M=S^{T}DS$ (\ref{Williamson}) of $M$ the condition $\Sigma +\frac{i\hbar }{2}%
J=M^{-1}+iJ\geq 0$ implies that $D^{-1}+iJ\geq 0$, that is%
\begin{equation*}
D^{-1}+iJ=%
\begin{pmatrix}
(\Lambda ^{\sigma })^{-1} & iI_{n\times n} \\ 
-iI_{n\times n} & (\Lambda ^{\sigma })^{-1}%
\end{pmatrix}%
\geq 0
\end{equation*}%
where $\Lambda ^{\sigma }=\limfunc{diag}(\lambda _{1}^{\sigma },...,\lambda
_{n}^{\sigma })$ the $\lambda _{j}^{\sigma }$ being the symplectic
eigenvalues of $M$. The eigenvalues of $D^{-1}+iJ$ are the real numbers $%
\lambda _{j}=(\lambda _{j}^{\sigma })^{-1}\pm 1$ and the condition $%
D^{-1}+iJ\geq 0$ thus implies that we must have $\lambda _{j}^{\sigma }\leq
1 $ for $1\leq j\leq n$. It follows that the ellipsoid $\Omega _{D}:Dz\cdot
z\leq \hbar $ contains the ball $B^{2n}(\sqrt{\hbar })$ and hence $\Omega
_{M}$ contains the quantum blob $S(B^{2n}(\sqrt{\hbar }))$ where $S\in 
\limfunc{Sp}(n)$ is the diagonalizing matrix. The result now follows
applying Propositions \ref{Thm2}.
\end{proof}

\begin{remark}
Assume that the Wigner distribution of $\widehat{\rho }$ is a Gaussian (\ref%
{Gaussian}) with $z_{0}=0$. If $\Omega ^{\hbar ,\sigma }=\Omega $ it follows
from Propositions \ref{Thm2}) that $\Omega $ is a quantum blob $S(B^{2n}(%
\sqrt{\hbar }))$ and hence 
\begin{equation}
\rho (z)=\frac{1}{(2\pi )^{n}}e^{-\frac{1}{2}Gz\cdot z}\text{ \ , \ }G=S^{T}S
\label{gaussmix}
\end{equation}%
for some $S\in \limfunc{Sp}(n)$. Then \cite{Birk,Birkbis} $\rho =W(\widehat{S%
}^{-1}\phi _{0}^{\hbar })$ where 
\begin{equation*}
\phi _{0}^{\hbar }(x)=(\pi \hbar )^{-n}e^{-|x|^{2}/2\hbar }
\end{equation*}
and $\widehat{S}\in \limfunc{Mp}(n)$ has projection $\pi ^{\limfunc{Mp}}(%
\widehat{S})=S$.
\end{remark}

\subsubsection{A dynamical characterization of admissibility}

Let $\Omega $ be an ellipsoid in $\mathbb{R}^{2n}$ with smooth boundary $%
\partial \Omega $. We assume that $\partial \Omega $ is the energy
hypersurface of some (quadratic) Hamiltonian function $H\in C^{\infty }(%
\mathbb{R}^{2n},\mathbb{R})$: \textit{i.e}. $\partial \Omega =\{z:H(z)=E\}$
for some $E\in \mathbb{R}$. We ask now when $\Omega $ can be viewed as the
covariance ellipsoid of a quantum state; the following is in a sense a
restatement of Theorem \ref{Thm2}, but we give an independent proof here:

\begin{theorem}
\label{ThNarcow}The ellipsoid $\Omega $ \ is a quantum covariance ellipsoid $%
\Omega _{\mathrm{cov}}$ (resp. an information ellipsoid $\Omega _{\mathrm{%
\inf o}}$) if and only if the following equivalent conditions are satisfied:
(i) We have 
\begin{equation}
\int_{\gamma }pdx\geq c(\Omega )\geq \pi \hbar  \label{covhz}
\end{equation}%
for every periodic Hamiltonian orbit $\gamma $ carried by $\partial \Omega $%
; if we have equality for the shortest orbit then $\Omega $ is a quantum
blob. (ii) Let $\Omega ^{\ast }$ be the Legendre transform of $\Omega $; we
have 
\begin{equation}
c(\Omega ^{\ast })\leq 4\pi /\hbar .  \label{comegastar}
\end{equation}%
(iii) Let $F$ be an arbitrary two-dimensional subspace of $\mathbb{R}^{2n}$
an let the ellipse $\gamma _{F}^{\ast }=\partial \Omega ^{\ast }\cap F$ be
positively oriented. We have 
\begin{equation}
\int_{\gamma _{F}^{\ast }}pdx\leq \frac{4\pi }{\hbar }.  \label{covnarcow}
\end{equation}
\end{theorem}

\begin{proof}
(i) The inequality (\ref{covhz}) follows from (\ref{HZ}) and (\ref%
{capellipse}). (ii) If $\Omega $ is defined by $\frac{1}{2}\Sigma
^{-1}z\cdot z\leq 1$ then $\Omega ^{\ast }$ is defined by $\frac{1}{2}\Sigma
z\cdot z\leq 1$. Setting $M=\frac{\hbar }{2}\Sigma ^{-1}$ the ellipsoid $%
\Omega ^{\ast }$ is given by $Nz\cdot z\leq \hbar $ where $N=\frac{\hbar ^{2}%
}{4}M^{-1}$. The symplectic spectrum of $N$ is thus $\frac{\hbar ^{2}}{4}%
((\lambda _{n}^{\sigma })^{-1},...,\lambda _{1}^{\sigma })^{-1})$ where $%
(\lambda _{1}^{\sigma },...,\lambda _{n}^{\sigma })$ is the symplectic
spectrum of $M$ (recall our convention to rank symplectic eigenvalues in
non-increasing order). It follows that the symplectic capacity of $\Omega
^{\ast }$ is 
\begin{equation*}
c(\Omega ^{\ast })=\pi \hbar (\frac{\hbar ^{2}}{4}(\lambda _{n}^{\sigma
})^{-1})^{-1}\geq \frac{4\pi }{\hbar }
\end{equation*}%
the last inequality because $\Omega $ is quantum admissible if and only if $%
\lambda _{n}^{\sigma }=\lambda _{\max }^{\sigma }\leq 1$. To prove the
action inequality (\ref{covnarcow}) we can proceed as follows: suppose first
that $F$ is a null space for the symplectic form (i.e. $F$ has a basis $%
\{e_{1},e_{2}\}$ such that $\sigma (e_{1},e_{2})=0$). The, by Stokes's
theorem we have 
\begin{equation}
\int_{\gamma _{F}^{\ast }}pdx=\int_{_{\Omega ^{\ast }\cap F}}\sigma =0
\end{equation}%
so that (\ref{covnarcow}) is trivially verified. Assume next that $F$ is a
symplectic plane. Then, by formula (\ref{areadualF}) in Theorem \ref{Thm2}
we have%
\begin{equation}
\func{Area}(\Omega _{M}^{\hbar ,\sigma }\cap F)\leq \pi \hbar
\end{equation}%
but this is precisely (\ref{covnarcow}) since $M=\frac{\hbar }{2}\Sigma
^{-1} $.
\end{proof}

\part{Lagrangian Polar Duality and Geometric Quantum States}

\section{Lagrangian polar duality and frames}

\subsection{Definition;\ Lagrangian frames}

In the Introduction we defined a notion of Lagrangian polar duality with
respect to a pair 
\begin{equation*}
(\ell ,\ell ^{\prime })\in \limfunc{Lag}\nolimits^{2}(n)=\limfunc{Lag}%
(n)\times \limfunc{Lag}(n)
\end{equation*}%
of Lagrangian planes in $(\mathbb{R}^{2n},\sigma )$ as follows: if $X_{\ell
} $ is a convex body contained in $\ell $ then its Lagrangian polar dual $%
(X_{\ell })_{\ell ^{\prime }}^{\hslash }$ with respect to $\ell ^{\prime }$
is defined as the set%
\begin{equation}
(X_{\ell })_{\ell ^{\prime }}^{\hbar }=\{z^{\prime }\in \ell ^{\prime
}:\sup\nolimits_{z\in \ell }\sigma (z,z^{\prime })\leq \hbar \}.
\label{lapodu}
\end{equation}%
This definition can be seen as the restriction of the symplectic polar
duality (\ref{sypodu}) studied above to the subset $\limfunc{Lag}^{2}(n)$ of 
$\mathbb{R}^{2n}$. For reasons that will become clear in a moment, we will
demand that the Lagrangian planes $\ell $ and $\ell ^{\prime }$ in this
definition be transversal,\textit{\ i.e.} that $\ell \cap \ell ^{\prime }=0$%
, that is, equivalently, $\ell \oplus \ell ^{\prime }=\mathbb{R}^{2n}$. Such
a pair of Lagrangian planes will be called a \emph{Lagrangian frame}. We
denote by $\limfunc{Lag}\nolimits_{0}^{2}(n)$ the set of all Lagrangian
frames:%
\begin{equation*}
\limfunc{Lag}\nolimits_{0}^{2}(n)=\{(\ell ,\ell ^{\prime })\in \limfunc{Lag}%
\nolimits^{2}(n):\ell \cap \ell ^{\prime }=0\}.
\end{equation*}%
Introducing the notation $\ell _{X}=\mathbb{R}_{x}^{n}\times 0$ \ and $\ell
_{P}=0\times \mathbb{R}_{p}^{n}$ we will call $(\ell _{X},\ell _{P})$ the 
\emph{canonical Lagrangian frame}. In fact, every Lagrangian frame can be
obtained from the canonical one using a linear symplectic automorphism. To
see this, we begin by noticing that the natural (transitive ) action 
\begin{equation}
\limfunc{Sp}(n)\times \limfunc{Lag}(n)\ni (S,\ell )\longmapsto S\ell \in 
\limfunc{Lag}(n)  \label{transac1}
\end{equation}%
induces a natural action 
\begin{equation}
\limfunc{Sp}(n)\times \limfunc{Lag}\nolimits_{0}^{2}(n)\longrightarrow 
\limfunc{Lag}\nolimits_{0}^{2}(n)  \label{transac2}
\end{equation}%
which is also transitive \cite{Birk}:

\begin{lemma}
\label{LemmaSp}The natural action%
\begin{equation*}
\limfunc{Sp}(n)\times \limfunc{Lag}\nolimits_{0}^{2}(n)\ni (S,(\ell ,\ell
^{\prime }))\longmapsto (S\ell ,S\ell ^{\prime })\in \limfunc{Lag}%
\nolimits_{0}^{2}(n)
\end{equation*}%
is transitive. In particular, for every Lagrangian frame $(\ell ,\ell
^{\prime })$ there exists $S\in \limfunc{Sp}(n)$ such that $(\ell ,\ell
^{\prime })=S(\ell _{X},\ell _{P})$ where $(\ell _{X},\ell _{P})$ is the
canonical Lagrangian frame.
\end{lemma}

\begin{proof}
Choose a basis $(e_{1i})_{1\leq 1\leq n}$ of $\ell _{1}$ and a basis $%
(f_{1j})_{1\leq j\leq n}$ of $\ell _{1}^{\prime }$ whose union $%
(e_{1i})_{1\leq 1\leq n}\cup (f_{1j})_{1\leq j\leq n}$ is a symplectic basis
of $(\mathbb{R}_{z}^{2n},\sigma )$. Similarly choose bases $(e_{2i})_{1\leq
1\leq n}$ and $f_{2j})_{1\leq j\leq n}$ of $\ell _{2}$ and $\ell
_{2}^{\prime }$ whose union is also a symplectic basis. The linear
automorphism of $\mathbb{R}^{2n}$ defined by $S(e_{1i})=e_{2i}$ and $%
S(f_{1i})=f_{2i}$ for $1\leq i\leq n$ is in $\limfunc{Sp}(n)$ and we have $%
(\ell _{2},\ell _{2}^{\prime })=(S\ell _{1},S\ell _{1}^{\prime })$.
\end{proof}

\begin{remark}
The transitivity result above can be extended to the set%
\begin{equation*}
\limfunc{Lag}\nolimits_{k}^{2}(n)=\{(\ell ,\ell ^{\prime })\in \limfunc{Lag}%
\nolimits^{2}(n):\dim (\ell \cap \ell ^{\prime })=k\}
\end{equation*}%
with $0\leq k\leq n$ by showing that the action $\limfunc{Sp}(n)\times 
\limfunc{Lag}\nolimits_{k}^{2}(n)\longmapsto \limfunc{Lag}%
\nolimits_{k}^{2}(n)$ is transitive as well \cite{Birk}.
\end{remark}

Notice that the symplectic automorphisms $S$ taking a Lagrangian frame $%
(\ell ,\ell ^{\prime })$ to another, $(\ell ^{\prime \prime },\ell ^{\prime
\prime \prime })$ is not unique. Let for instance $S,S^{\prime }\in \limfunc{%
Sp}(n)$ be such that 
\begin{equation*}
(\ell ,\ell ^{\prime })=S(\ell _{X},\ell _{P})=S^{\prime }(\ell _{X},\ell
_{P}).
\end{equation*}%
Then $(S^{\prime })^{-1}S(\ell _{X},\ell _{P})=(\ell _{X},\ell _{P})$ which
implies that $(S^{\prime })^{-1}S$ is in the isotropy subgroups of both $%
\ell _{X}$ and $\ell _{P}$, that is, 
\begin{equation*}
(S^{\prime })^{-1}S=M_{L}=%
\begin{pmatrix}
L^{-1} & 0_{n\times n} \\ 
0_{n\times n} & L^{T}%
\end{pmatrix}%
\end{equation*}%
for some $L\in GL(n,\mathbb{R})$. \ (The symplectic matrices $M_{L}$
correspond to the embedding $GL(n,\mathbb{R})\hookrightarrow \limfunc{Sp}(n)$
whose lift to the metaplectic group $\limfunc{Mp}(n)$ is the metalinear
group $\limfunc{ML}(n)$; see \cite{GS} for a discussion of metalinear
structures).

\subsection{John--L\"{o}wner ellipsoids and Lagrangian polar duality}

The notion of Mahler volume generalizes without difficulty to Lagrangian
polar duality. Defining 
\begin{equation}
v_{\limfunc{Lag}}(X_{\ell })=\limfunc{Vol}\nolimits_{2n}(X_{\ell }\times
X_{\ell ^{\prime }}^{\hbar })  \label{MahlerLag}
\end{equation}%
we have $v_{\limfunc{Lag}}(X_{\ell })=v(X)$ if $X_{\ell }\times X_{\ell
^{\prime }}^{\hbar }=S(X\times X^{\hbar })$ for $S\in \limfunc{Sp}(n)$ since
symplectic automorphisms are volume preserving.

The usual duality relations (\ref{JL}) between John and L\"{o}wner
ellipsoids readily extend to Lagrangian polar duality: let $(\ell ,\ell
^{\prime })$ be a Lagrangian frame and $X_{\ell }$ a convex body centered at
the origin; then 
\begin{equation}
((X_{\ell })_{\mathrm{John}})_{\ell ^{\prime }}^{\hbar }=((X_{\ell })_{\ell
^{\prime }}^{\hbar })_{\mathrm{L\ddot{o}wner}}\text{ \ },\text{ \ }((X_{\ell
})_{\mathrm{L\ddot{o}wner}})_{\ell ^{\prime }}^{\hbar }=((X_{\ell })_{\ell
^{\prime }}^{\hbar })_{\mathrm{John}}.  \label{JLag}
\end{equation}%
The following elementary result is very important for the definition of the
geometric quantum states we give below:

\begin{lemma}
\label{LemmaJohn}Let $R>0.$ The John and L\"{o}wner ellipsoids of $%
B_{X}^{n}(R)\times B_{P}^{n}(R)$ are, respectively,%
\begin{eqnarray}
(B_{X}^{n}(R)\times B_{P}^{n}(R))_{\mathrm{John}} &=&B^{2n}(R)  \label{John}
\\
(B_{X}^{n}(R)\times B_{P}^{n}(R))_{\mathrm{L\ddot{o}wner}} &=&B^{2n}(2R).
\label{Loewner}
\end{eqnarray}
\end{lemma}

\begin{proof}
The inclusion 
\begin{equation}
B^{2n}(R)\subset B_{X}^{n}(R)\times B_{P}^{n}(R)  \label{incl}
\end{equation}%
is obvious, and we cannot have 
\begin{equation*}
B^{2n}(R^{\prime })\subset B_{X}^{n}(R)\times B_{P}^{n}(R)
\end{equation*}%
if $R^{\prime }>R$. Assume now that the John ellipsoid $\Omega _{\mathrm{John%
}}$ of $\Omega =B_{X}^{n}(R)\times B_{P}^{n}(R)$ is defined by $%
Ax^{2}+Bxp+Cp^{2}\leq R^{2}$ where $A,C>0$ and $B$ are real $n\times n$
matrices. Since $\Omega $ is invariant by the transformation $%
(x,p)\longmapsto (p,x)$ so is $\Omega _{\mathrm{John}}$ and we must thus
have $A=C$ and $B=B^{T}$. Similarly, $\Omega $ being invariant by the
partial reflection $(x,p)\longmapsto (-x,p)$ we get $B=0$ so $\Omega _{%
\mathrm{John}}$ is defined by $Ax^{2}+Ap^{2}\leq R^{2}$. We next observe
that $\Omega $ and hence $\Omega _{\mathrm{John}}$ are invariant under the
transformations $(x,p)\longmapsto (Hx,HP)$ where $H\in O(n,\mathbb{R})$ so
we must have $AH=HA$ for all $H\in O(n,\mathbb{R})$, but this is only
possible if $A=\lambda I_{n\times n}$ for some $\lambda \in \mathbb{R}$. The
John ellipsoid is thus of the type $B^{2n}(R/\sqrt{\lambda })$ for some $%
\lambda \geq 1$ and this concludes the proof in view of the inclusion (\ref%
{incl}) since the case $\lambda >R^{2}$ is excluded. Formula (\ref{Loewner})
for the L\"{o}wner ellipsoid is proven in a similar way.
\end{proof}

\subsection{Geometric quantum states\label{secpure}}

\subsubsection{Elliptic geometric states}

In \cite{MCFOOP} we gave the following definition:

\begin{definition}
\label{DefGeom}Let $(\ell ,\ell ^{\prime })\in \limfunc{Lag}_{0}^{2}(n)$ be
a Lagrangian frame in $(\mathbb{R}^{2n},\sigma )$ and let $X_{\ell }\subset
\ell $ be an ellipsoid with center $0$. We call the product $X_{\ell }\times
(X_{\ell })_{\ell ^{\prime }}^{\hbar }\subset \mathbb{R}^{2n}$ the geometric
quantum state in $\mathbb{R}^{2n}$ associated with the frame $(\ell ,\ell
^{\prime })$ and the ellipsoid $X_{\ell }$. We denote by $\limfunc{Quant}%
\nolimits_{0}^{\mathrm{Ell}}(n)$ the set of all such centered geometric
states on $\mathbb{R}^{2n}$.
\end{definition}

The simplest example of a geometric quantum state in $\mathbb{R}^{2n}$
associated with the canonical Lagrangian frame $(\ell _{X},\ell _{P})$ and
the ball $B_{X}^{n}(\sqrt{\hbar })\subset \ell _{X}$ is%
\begin{equation}
X_{\ell _{X}}\times (X_{\ell _{X}})_{\ell _{P}}^{\hbar }=B_{X}^{n}(\sqrt{%
\hbar })\times B_{P}^{n}(\sqrt{\hbar })  \label{fidu}
\end{equation}%
as follows from the identity $B_{X}^{n}(\sqrt{\hbar })^{\hbar }=B_{P}^{n}(%
\sqrt{\hbar })$. We will call it the \textquotedblleft standard geometric
state\textquotedblright .

\subsubsection{Symplectic actions}

Lemma \ref{LemmaSp} allows us to reduce the study of Lagrangian polar
duality to that of ordinary polar duality; it also allows us to prove the
important result for the action 
\begin{equation}
X_{\ell }\times (X_{\ell })_{\ell ^{\prime }}^{\hbar })\longmapsto S(X_{\ell
}\times (X_{\ell })_{\ell ^{\prime }}^{\hbar }))  \label{action1}
\end{equation}%
of $\limfunc{Sp}(n)$ on geometric states. Let us first introduce some
notation. Let $X_{\ell }\subset \ell $ be a centered convex body with
Lagrangian polar dual $(X_{\ell })_{\ell ^{\prime }}^{\hbar }\subset \ell
^{\prime }$. For $S\in \limfunc{Sp}(n)$ we define $Y_{S\ell }=S(X_{\ell })$.
It is a centered convex body carried by the Lagrangian plane $S\ell $.

\begin{lemma}
Let $(\ell ,\ell ^{\prime })$ be a Lagrangian frame and $S\in \limfunc{Sp}%
(n) $. The Lagrangian polar dual of $Y_{S\ell }=S(X_{\ell })$ with respect
to $S\ell ^{\prime }$ is 
\begin{equation}
(Y_{S\ell })_{S\ell ^{\prime }}^{\hbar }=S\left[ (X_{\ell })_{\ell ^{\prime
}}^{\hbar }\right] =\left[ S(X_{\ell })\right] _{S\ell ^{\prime }}^{\hbar }.
\label{sx}
\end{equation}%
and $\limfunc{Sp}(n)$ thus acts on elliptic geometric states via the rule%
\begin{equation}
S(X_{\ell }\times (X_{\ell })_{\ell ^{\prime }}^{\hbar })=(S(X_{\ell
})\times \left[ S(X_{\ell })\right] _{S\ell ^{\prime }}^{\hbar })
\label{sxlrule}
\end{equation}
\end{lemma}

\begin{proof}
Let $z\in S\left[ (X_{\ell })_{\ell ^{\prime }}^{\hbar }\right] $, that is $%
S^{-1}z\in (X_{\ell })_{\ell ^{\prime }}^{\hbar }$. This is equivalent to
the conditions $z\in S\ell ^{\prime }$ and $\sigma (S^{-1}z,z^{\prime })\leq
\hbar $ for all $z^{\prime }\in X_{\ell }$. Since $\sigma (S^{-1}z,z^{\prime
})=\sigma (z,Sz^{\prime })$ this is in turn equivalent to $z\in S\ell
^{\prime }$ and $\sigma (z,Sz^{\prime })\leq \hbar $ for all $Sz^{\prime
}\in S(X_{\ell })$, that is to $z\in \left[ S(X_{\ell })\right] _{S\ell
^{\prime }}^{\hbar }$, establishing the second equality (\ref{sx}).
\end{proof}

Applying this result to geometric quantum states we get:

\begin{proposition}
\label{propsycogeom}Let $(\ell ,\ell ^{\prime })$ be a Lagrangian frame and
Let $S\in \limfunc{Sp}(n)$ be such that $(\ell ,\ell ^{\prime })=S(\ell
_{X},\ell _{P})$ and set $X=S^{-1}(X_{\ell })\subset \ell _{X}$. (i) We have 
$X^{\hbar }=S^{-1}(X_{\ell })_{\ell ^{\prime }}^{\hbar }$ $\subset \ell _{P}$%
, that is 
\begin{equation}
X_{\ell }\times (X_{\ell })_{\ell ^{\prime }}^{\hbar }=S(X)\times S(X^{\hbar
})=S(X\times X^{\hbar }).  \label{sxl}
\end{equation}%
(ii) Let $(\ell ^{\prime \prime },\ell ^{\prime \prime \prime })=S^{\prime
}(\ell ,\ell ^{\prime })$, $S^{\prime }\in \limfunc{Sp}(n)$, be a second
Lagrangian frame, and $X_{\ell ^{\prime \prime }}\subset \ell ^{\prime
\prime }$ a convex body. Setting $X_{\ell }=(S^{\prime })^{-1}(X_{\ell
^{\prime \prime }})$ we have 
\begin{equation}
X_{\ell ^{\prime \prime }}\times (X_{\ell ^{\prime \prime }})_{\ell ^{\prime
\prime \prime }}^{\hbar }=S^{\prime }(X_{\ell })\times S^{\prime }(X_{\ell
})_{\ell ^{\prime }}^{\hbar }=S^{\prime }(X_{\ell }\times (X_{\ell })_{\ell
^{\prime }}^{\hbar })  \label{sxlh}
\end{equation}%
and the action (\ref{action1}) defined above is thus transitive.
\end{proposition}

\begin{proof}
(i) It suffices to prove that $S^{-1}((X_{\ell })_{\ell ^{\prime }}^{\hbar
})=X^{\hbar }$. The condition $z^{\prime }\in S^{-1}((X_{\ell })_{\ell
^{\prime }}^{\hbar })$ is equivalent to $Sz^{\prime }\in (X_{\ell })_{\ell
^{\prime }}^{\hbar }$, that is, to $\sigma (Sz^{\prime },z)=\sigma
(z^{\prime },S^{-1}z)\leq \hbar $ for all $z\in X_{\ell }$, which we can
rewrite as $\sigma (z^{\prime },z^{\prime \prime })\leq \hbar $ for all $%
z^{\prime \prime }\in S^{-1}(X_{\ell })=X$, hence $z^{\prime }\in X^{\hbar }$%
. Formula (\ref{sxl}) follows since we have just shown that $S^{-1}(X_{\ell
},X_{\ell ^{\prime }}^{\hbar })=(X,X^{\hbar })$. \ Property (ii) immediately
follows from (i).
\end{proof}

In view of Proposition \ref{propsycogeom} the natural symplectic action%
\begin{gather}
\limfunc{Sp}(n)\times \limfunc{Quant}\nolimits_{0}^{\mathrm{Ell}%
}(n)\longrightarrow \limfunc{Quant}\nolimits_{0}^{\mathrm{Ell}}(n)
\label{trans1} \\
(S\,,\,X_{\ell }\times (X_{\ell })_{\ell ^{\prime }}^{\hbar })\longmapsto
S\left( X_{\ell }\times (X_{\ell })_{\ell ^{\prime }}^{\hbar })\right)
\label{trans2}
\end{gather}%
is transitive. Explicitly this action is described as follows: if $(\ell
,\ell ^{\prime })=S_{0}(\ell _{X},\ell _{P})$ and $X=S_{0}^{-1}(X_{\ell
})\subset \ell _{X}$. and $X^{\hbar }=S_{0}^{-1}(X_{\ell })_{\ell ^{\prime
}}^{\hbar }\subset \ell _{P}$ then 
\begin{equation}
S(X_{\ell }\times (X_{\ell })_{\ell ^{\prime }}^{\hbar }))=SS_{0}(X)\times
SS_{0}(X^{\hbar })=X_{S\ell }\times (X_{S\ell })_{S\ell ^{\prime }}^{\hbar }.
\label{trans2bis}
\end{equation}

We will use several times below the symplectic rescaling matrix%
\begin{equation*}
M_{L}=%
\begin{pmatrix}
L^{-1} & 0_{n\times n} \\ 
0_{n\times n} & L^{T}%
\end{pmatrix}%
\in \limfunc{Sp}(n)
\end{equation*}%
where $L\in GL(n,\mathbb{R})$; note that $(M_{L})^{-1}=M_{L^{-1}}$.

Every centered geometric quantum state associated with and ellipsoid can be
obtained from the standard state by a symplectic automorphism; in fact part
(ii) of the proposition below justifies the terminology \textquotedblleft
geometric quantum state\textquotedblright\ for products $X_{\ell }\times
(X_{\ell })_{\ell ^{\prime }}^{\hbar }$.

\begin{proposition}
\label{PropStandard}Let $(\ell ,\ell ^{\prime })$ be a Lagrangian frame and $%
X_{\ell }\subset \ell $ a centered ellipsoid. There exists $S\in \limfunc{Sp}%
(n)$ such that 
\begin{equation}
X_{\ell }\times (X_{\ell })_{\ell ^{\prime }}^{\hbar }=S(B_{X}^{n}(\sqrt{%
\hbar })\times B_{P}^{n}(\sqrt{\hbar })).  \label{xxlb}
\end{equation}%
(ii) The John ellipsoid $(X_{\ell }\times (X_{\ell })_{\ell ^{\prime
}}^{\hbar })_{\mathrm{John}}$ of the geometric quantum state $X_{\ell
}\times (X_{\ell })_{\ell ^{\prime }}^{\hbar }$ is a quantum blob $%
Q_{S^{\prime }}=S^{\prime }B^{2n}(\sqrt{\hbar }))$, $S^{\prime }\in \limfunc{%
Sp}(n)$.
\end{proposition}

\begin{proof}
(i) Let $S^{\prime }\in \limfunc{Sp}(n)$ be such that $(\ell ,\ell ^{\prime
})=S^{\prime }(\ell _{X},\ell _{P})$. Then $X=S^{-1}(X_{\ell })$ is a
centered ellipsoid in $\ell _{X}$ and there exists $A\in GL(n,\mathbb{R})$
such that $X=A(B_{X}^{n}(\sqrt{\hbar }))$ and $X^{\hbar
}=(A^{T})^{-1}(B_{P}^{n}(\sqrt{\hbar })).$ Hence%
\begin{equation}
X_{\ell }\times (X_{\ell })_{\ell ^{\prime }}^{\hbar }=S^{\prime
}M_{A^{-1}}(B_{X}^{n}(\sqrt{\hbar })\times B_{P}^{n}(\sqrt{\hbar }))
\label{sma}
\end{equation}%
and (\ref{xxlb}) holds with $S=S^{\prime }M_{A^{-1}}$. Property (ii) follows
from Lemma \ref{LemmaJohn} with $R=\sqrt{\hbar }$ using formula (\ref{sma}).
\end{proof}

We next notice that the standard geometric state $B_{X}^{n}(\sqrt{\hbar }%
)\times B_{P}^{n}(\sqrt{\hbar })$ is invariant\ by the action of the
subgroup $O(n)$ of $\limfunc{Sp}(n)$ consisting of all matrices $M_{H}=%
\begin{pmatrix}
H & 0_{n\times n} \\ 
0_{n\times n} & H%
\end{pmatrix}%
$ with $H\in O(n,\mathbb{R})$. The latter is actually a subgroup of the
group of symplectic rotations $U(n)=\limfunc{Sp}(n)\cap O(2n,\mathbb{R})$
which is in turn identified with the unitary group $U(n,\mathbb{C})$ using
the canonical monomorphism 
\begin{equation*}
\iota :u=A+iB\longmapsto U=%
\begin{pmatrix}
A & B \\ 
-B & A%
\end{pmatrix}%
.
\end{equation*}%
Note that $\iota (u^{\ast })=\iota (u)^{T}$ hence the relation $uu^{\ast
}=u^{\ast }u=I_{n\times n}$ in $U(n,\mathbb{C})$ becomes $%
UU^{T}=U^{T}U=I_{2n\times 2n}$ in $U(n)$. Since $O(n)$ is a closed subgroup
of $\limfunc{Sp}(n)$ it follows from the invariance of $B_{X}^{n}(\sqrt{%
\hbar })\times B_{P}^{n}(\sqrt{\hbar })$ under $O(n)$ that we have the
canonical identification%
\begin{equation}
\limfunc{Quant}\nolimits_{0}^{\mathrm{Ell}}(n)\equiv \limfunc{Sp}(n)/O(n).
\label{quanton}
\end{equation}

\subsubsection{Geometric quantum states with an arbitrary center}

Until now we have assumed that the set $X_{\ell }$ and its Lagrangian polar
dual $X_{\ell ^{\prime }}^{\hbar }$ were centered at the origin. The general
case of elliptic geometric quantum states is easily defined using
translations. We denote by $T(z_{0})$ the mapping $z\longmapsto z+z_{0}$.
Let $\ell \in \limfunc{Lag}(n)$ and $z_{0}\in \mathbb{R}^{2n}$. We set $\ell
(z_{0})=T(z_{0})\ell =\ell +z_{0}$. If $(\ell ,\ell ^{\prime })\in \limfunc{%
Lag}_{0}^{2}(n)$ is a Lagrangian frame we will call $(\ell (z_{0}),\ell
^{\prime }(z_{0}))$ an \textit{affine Lagrangian frame}. Let $X_{\ell
(z_{0})}$ be an ellipsoid centered at $z_{0}$ and carried by $\ell (z_{0})$;
the set $X_{\ell }=T(-z_{0})X_{\ell (z_{0})}$ is an ellipsoid with center $0$
and carried by $\ell $, this \textit{a priori} motivates the notation $%
X_{\ell (z_{0})}=X_{\ell }(z_{0})$ which is consistent with the notation $%
\ell (z_{0})=T(z_{0})\ell $; we will however in general avoid it because it
can lead to ambiguities in some calculations.

\begin{definition}
\label{Defquant}Let $(\ell (z_{0}),\ell ^{\prime }(z_{0}))$ be an \textit{%
affine Lagrangian frame associated with }$(\ell ,\ell ^{\prime })\in 
\limfunc{Lag}_{0}^{2}(n)$. We define the polar dual of $X_{\ell
(z_{0})}=z_{0}+X_{\ell }$ with respect to $\ell ^{\prime }(z_{0})$ by 
\begin{equation*}
(X_{\ell (z_{0})})_{\ell ^{\prime }(z_{0})}^{\hslash }=T(z_{0})(X_{\ell
})_{\ell ^{\prime }}^{\hslash }=z_{0}+(X_{\ell })_{\ell ^{\prime }}^{\hslash
}.
\end{equation*}%
The (affine) geometric quantum state associated with $(\ell (z_{0}),\ell
^{\prime }(z_{0}))$and $X_{\ell (z_{0})}$ is the Cartesian product 
\begin{equation}
X_{\ell (z_{0})}\times (X_{\ell (z_{0})})_{\ell ^{\prime }(z_{0})}^{\hslash
}=(z_{0}+X_{\ell })\times (z_{0}+(X_{\ell })_{\ell ^{\prime }}^{\hslash }).
\label{defgeneral}
\end{equation}%
We denote $\limfunc{Quant}\nolimits^{\mathrm{Ell}}(n)$ the set of all such
products.
\end{definition}

We have the obvious inclusion $\limfunc{Quant}\nolimits_{0}^{\mathrm{Ell}%
}(n)\subset \limfunc{Quant}\nolimits^{\mathrm{Ell}}(n)$; the properties of
centered geometric states studied above carry over to this more general case
without major difficulties. For instance, the action of $S\in \limfunc{Sp}%
(n) $ on $X_{\ell (z_{0})}=z_{0}+X_{\ell }$ and its polar dual $(X_{\ell
(z_{0})})_{\ell ^{\prime }(z_{0})}^{\hslash }$ is given by the formulas (%
\textit{cf}. (\ref{sx})) 
\begin{eqnarray*}
S(X_{\ell (z_{0})}) &=&Sz_{0}+S(X_{\ell }) \\
S\left[ (X_{\ell (z_{0})})_{\ell ^{\prime }(z_{0})}^{\hslash }\right]
&=&Sz_{0}+\left[ S(X_{\ell })\right] _{S\ell ^{\prime }}^{\hbar }.
\end{eqnarray*}%
The symplectic action on centered quantum states described by (\ref{trans1}%
)--(\ref{trans2})--(\ref{trans2bis}) thus induces an action%
\begin{equation}
\limfunc{Sp}(n)\times \limfunc{Quant}\nolimits^{\mathrm{Ell}%
}(n)\longrightarrow \limfunc{Quant}\nolimits^{\mathrm{Ell}}(n)
\label{trans3}
\end{equation}%
defined by 
\begin{equation}
S\left( X_{\ell (z_{0})}\times (X_{\ell (z_{0})})_{\ell ^{\prime
}(z_{0})}^{\hslash }\right) =\left( Sz_{0}+X_{S\ell }\right) \times \left(
Sz_{0}+(X_{S\ell }))_{S\ell ^{\prime }}^{\hbar }\right)  \label{trans4}
\end{equation}%
which we can write more compactly as%
\begin{equation}
S\left( X_{\ell (z_{0})}\times (X_{\ell (z_{0})})_{\ell ^{\prime
}(z_{0})}^{\hslash }\right) =X_{S\ell (Sz_{0})}\times (X_{S\ell
(Sz_{0})})_{S\ell ^{\prime }(Sz_{0})}^{\hslash }.  \label{trans4bis}
\end{equation}%
More generally, the inhomogeneous symplectic group $\func{ISp}(n)=\limfunc{Sp%
}(n)\ltimes \mathbb{R}^{2n}$ also acts on affine geometric quantum states:%
\begin{equation}
\func{ISp}(n)\times \limfunc{Quant}\nolimits^{\mathrm{Ell}%
}(n)\longrightarrow \limfunc{Quant}\nolimits^{\mathrm{Ell}}(n).
\label{trans5}
\end{equation}

\subsection{Geometric states and generalized Gaussians}

\subsubsection{Centered Gaussians}

Let $A$ and $B$ are real symmetric $n\times n$ matrices, $A$ positive
definite, and $\gamma \in \mathbb{R}$. We define a function $\psi
_{A,B}^{\gamma }\in \mathcal{S}(\mathbb{R}^{n})$ by $\psi _{A,B}^{\gamma
}=e^{i\gamma }\psi _{A,B}$ where $\psi _{A,B}$ is defined by formula (\ref%
{psix1}): 
\begin{equation}
\psi _{A,B}^{\gamma }(x)=e^{i\gamma }\left( \tfrac{1}{\pi \hbar }\right)
^{n/4}(\det A)^{1/4}e^{-\tfrac{1}{2\hbar }(A+iB)x\cdot x}.  \label{wpt}
\end{equation}%
This function is $L^{2}$-normalized: 
\begin{equation}
\text{\ }||\psi _{A,B}^{\gamma }||_{L^{2}(\mathbb{R}^{n})}=1.  \label{norm}
\end{equation}%
As a particular case we have $\psi _{I,0}^{0}=\phi _{0}^{\hbar }$, the
standard Gaussian defined by 
\begin{equation}
\phi _{0}^{\hbar }(x)=(\pi \hbar )^{-n/4}e^{-|x|^{2}/2\hbar }.  \label{fid}
\end{equation}

We denote by $\limfunc{Gauss}_{0}(n)$ the set of all equivalence classes of
Gaussians (\ref{wpt}).for the equivalence relation 
\begin{equation}
\psi _{A,B}^{\gamma }\sim \psi _{A^{\prime },B^{\prime }}^{\gamma ^{\prime
}}\Longleftrightarrow A=A^{\prime }\text{ \ and }B=B^{\prime }
\label{equiv1}
\end{equation}%
and we will write simply $\psi _{A,B}$ to denote the equivalence class of $%
\psi _{A,B}^{\gamma }$. Every $\psi _{A,B}\in \limfunc{Gauss}_{0}(n)$ can be
obtained from the standard Gaussian (\ref{fid}) using elementary metaplectic
transforms. Recall \cite{Birk} that the metaplectic group $\limfunc{Mp}(n)$
is generated by the modified Fourier transform%
\begin{equation}
\widehat{F}\psi (x)=(2\pi i\hbar )^{-n/2}\int_{\mathbb{R}^{n}}e^{-\frac{i}{%
\hbar }x\cdot x^{\prime }}\psi (x^{\prime })dx^{\prime }  \label{ftmp}
\end{equation}%
together with the local operators%
\begin{equation}
\widehat{M}_{L,m}\psi (x)=i^{m}\sqrt{|\det L|}\psi (x)\text{ \ },\text{ \ }%
\widehat{V}_{-P}\psi (x)=e^{\frac{i}{2\hbar }Px\cdot x}\psi (x)  \label{mlvp}
\end{equation}%
where $L\in GL(n,\mathbb{R})$, $P\in \limfunc{Sym}(n,\mathbb{R})$ and the
integer $m$ corresponds to a choice of $\arg \det (L)$ (\textquotedblleft
Maslov index\textquotedblright\ \cite{Birk}). The projections of these
operators on $\limfunc{Sp}(n)$ are 
\begin{equation}
\pi ^{\limfunc{Mp}}(\widehat{F})=J,\text{ \ \ }\pi ^{\limfunc{Mp}}(\widehat{M%
}_{L,m})=M_{L},\text{ \ }\pi ^{\limfunc{Mp}}(\widehat{V}_{-P})=V_{-P}
\label{pimp}
\end{equation}
where 
\begin{equation}
M_{L}=%
\begin{pmatrix}
L^{-1} & 0_{n\times n} \\ 
0_{n\times n} & L^{T}%
\end{pmatrix}%
\text{ }\ ,\text{ \ }V_{-P}=%
\begin{pmatrix}
I_{n\times n} & 0_{n\times n} \\ 
P & I_{n\times n}%
\end{pmatrix}%
.  \label{projmlvp}
\end{equation}

\begin{lemma}
\label{Lemmapsiab}Let $\psi _{A,B}\in \limfunc{Gauss}_{0}(n)$. We have 
\begin{equation}
\psi _{A,B}=\widehat{V}_{B}\widehat{M}_{A^{1/2},0}\phi _{0}^{\hbar }
\label{psiabfi}
\end{equation}%
where $\phi _{0}^{\hbar }$ is the standard Gaussian (\ref{fid}) and the
Wigner function of $\psi _{A,B}$ is given by \ 
\begin{equation}
W\psi _{AB}(z)=(\pi \hbar )^{-n}e^{-\tfrac{1}{\hbar }G_{AB}z\cdot z}
\label{phagauss}
\end{equation}%
where $G_{AB}$ is the positive definite symplectic matrix 
\begin{equation}
G_{AB}=(S_{AB}S_{AB}^{T})^{-1}\text{ \ },\text{ \ }S_{AB}=%
\begin{pmatrix}
A^{-1/2} & 0_{n\times n} \\ 
-BA^{-1/2} & A^{1/2}%
\end{pmatrix}%
.  \label{gaga}
\end{equation}
\end{lemma}

\begin{proof}
Formula (\ref{psiabfi}) is obvious. In view of the symplectic covariance
formula (\ref{syw}) for Wigner functions we have%
\begin{equation*}
W\psi _{A,B}=W\phi _{0}^{\hbar }\circ (M_{A^{-1/2}}V_{-B})=W\phi _{0}^{\hbar
}\circ S_{AB}^{-1}
\end{equation*}%
\ hence (\ref{phagauss}) since we have \cite{Birk} 
\begin{equation*}
W\phi _{0}^{\hbar }(z)=(\pi \hbar )^{-n}e^{-|z|^{2}/\hbar }.
\end{equation*}%
The formulas (\ref{gaga}) follow from (\ref{psix3}).
\end{proof}

More generally, we have a natural action%
\begin{equation}
\limfunc{Mp}(n)\times \limfunc{Gauss}\nolimits_{0}(n)\longrightarrow 
\limfunc{Gauss}\nolimits_{0}(n)  \label{mpgauss}
\end{equation}%
and this action is transitive in view of (\ref{psiabfi}). Let $\widehat{S}%
\in \limfunc{Mp}(n)$; we have $\widehat{S}\psi _{A,B}=\widehat{S}^{\prime
}\phi _{0}^{\hbar }$ where $\widehat{S}^{\prime }=\widehat{S}\widehat{V}_{B}%
\widehat{M}_{A^{1/2},0}$. Let 
\begin{equation}
S^{\prime }=M_{L}V_{-P}U\text{ \ },\text{ }U\in U(n),L=L^{T}>0)  \label{Iwa}
\end{equation}%
be the pre-Iwasawa factorization \cite{Arvind,ATMP} of $S^{\prime }=\pi ^{%
\limfunc{Mp}}(\widehat{S}^{\prime })$. We have, by symplectic covariance and
taking into account the fact that $U^{-1}\in U(n)$, 
\begin{equation*}
W(\widehat{S}^{\prime }\phi _{0}^{\hbar })(z)=W\phi _{0}^{\hbar
}(U^{-1}V_{P}M_{L^{-1}}z)=W\phi _{0}^{\hbar }(V_{P}M_{L^{-1}}z)
\end{equation*}%
and hence 
\begin{equation*}
\widehat{S}\psi _{A,B}=\widehat{S}^{\prime }\phi _{0}^{\hbar }=\widehat{M}%
_{L,0}\widehat{V}_{-P}\phi _{0}^{\hbar }\in \limfunc{Gauss}\nolimits_{0}(n)
\end{equation*}%
(See \cite{Birk} for a different approach using Fourier integrals).

The following theorem identifies the set $\limfunc{Gauss}\nolimits_{0}(n)$
of centered Gaussian states $\psi _{AB}^{\gamma }$ with $\limfunc{Quant}%
\nolimits_{0}^{\mathrm{Ell}}(n)$. Recall that $\phi _{0}^{\hbar }$ is the
standard Gaussian (\ref{fid}).

\begin{theorem}
\label{Thm1}The mapping 
\begin{equation}
\Psi :\limfunc{Quant}\nolimits_{0}^{\mathrm{Ell}}(n)\longrightarrow \limfunc{%
Gauss}\nolimits_{0}(n)  \label{gausson}
\end{equation}%
defined by 
\begin{equation}
\Psi (X_{\ell }\times (X_{\ell })_{\ell ^{\prime }}^{\hslash })=\widehat{S}%
\phi _{0}^{\hbar }  \label{fix}
\end{equation}%
where $\widehat{S}\in \limfunc{Mp}(n)$ covers $S\in \limfunc{Sp}(n)$ such
that 
\begin{equation}
(X_{\ell }\times (X_{\ell })_{\ell ^{\prime }}^{\hslash })=S(B_{X}^{n}(\sqrt{%
\hbar })\times B_{P}^{n}(\sqrt{\hbar }))  \label{bxbp}
\end{equation}%
is a bijection.
\end{theorem}

\begin{proof}
Let us first show that $\Psi $ is well-defined \textit{i.e}. that $\widehat{S%
}\phi _{0}^{\hbar }$ does not depend on the choice of $S$ in (\ref{bxbp}).
If $S^{\prime }\in \limfunc{Sp}(n)$ is such that 
\begin{equation*}
(X_{\ell }\times (X_{\ell })_{\ell ^{\prime }}^{\hslash })=S^{\prime
}(B_{X}^{n}(\sqrt{\hbar })\times B_{P}^{n}(\sqrt{\hbar }))
\end{equation*}%
then $(S^{\prime })^{-1}S$ must leave $B_{X}^{n}(\sqrt{\hbar })\times
B_{P}^{n}(\sqrt{\hbar })$ invariant which implies that we must have%
\begin{equation*}
(S^{\prime })^{-1}S=M_{R^{-1}}=%
\begin{pmatrix}
R & 0_{n\times n} \\ 
0_{n\times n} & R%
\end{pmatrix}%
\text{ \ , \ }R\in O(n,\mathbb{R})
\end{equation*}

and hence $S=S^{\prime }M_{R}$. It follows that $\widehat{S}=\widehat{%
S^{\prime }}\widehat{M}_{R^{-1},m}$ where $\widehat{S^{\prime }}$ covers $%
S^{\prime }$ and $\widehat{M}_{R^{-1},m}\psi (x)=i^{m}\psi (Rx)$ for some
integer $m$. We thus have 
\begin{equation*}
\widehat{S}\phi _{0}^{\hbar }=\widehat{S^{\prime }}\widehat{M}%
_{R^{-1},m}\phi _{0}^{\hbar }=\widehat{S^{\prime }}\phi _{0}^{\hbar }
\end{equation*}%
because $\widehat{M}_{R^{-1},m}\phi _{0}^{\hbar }=\phi _{0}^{\hbar }$ by
rotaional symmetry. The surjectivity of the mapping $\Psi $ follows from
Lemma \ref{Lemmapsiab}: for every $\psi \in \limfunc{Gauss}\nolimits_{0}(n)$
there exists $\widehat{S}\in \limfunc{Mp}(n)$ such that 
\begin{equation*}
\psi =\widehat{S}\phi _{0}^{\hbar }=\Psi \left( S(B_{X}^{n}(\sqrt{\hbar }%
)\times B_{P}^{n}(\sqrt{\hbar }))\right)
\end{equation*}%
where $S=\pi ^{\mathrm{Mp}}(\widehat{S})$. To prove that $\Psi $ is
injective one proceed as in the beginning of the proof: if $\widehat{S}\phi
_{0}^{\hbar }=\widehat{S^{\prime }}\phi _{0}^{\hbar }$ then $\widehat{S}=%
\widehat{S^{\prime }}\widehat{M}_{R,m}$ for some $R\in O(n,\mathbb{R})$
hence $S=S^{\prime }M_{R}$ so that 
\begin{equation*}
S(B_{X}^{n}(\sqrt{\hbar })\times B_{P}^{n}(\sqrt{\hbar }))=S^{\prime
}(B_{X}^{n}(\sqrt{\hbar })\times B_{P}^{n}(\sqrt{\hbar })).
\end{equation*}
\end{proof}

These generalization to the non-centered case is immediate: recalling that 
\begin{equation*}
S\left( X_{\ell (z_{0})}\times (X_{\ell (z_{0})})_{\ell ^{\prime
}(z_{0})}^{\hslash }\right) =\left( Sz_{0}+X_{S\ell }\right) \times \left(
Sz_{0}+(X_{S\ell }))_{S\ell ^{\prime }}^{\hbar }\right)
\end{equation*}%
the bijection \ $\Psi $defined by (\ref{fix}) extends to a bijection 
\begin{equation}
\Psi :\limfunc{Quant}\nolimits^{\mathrm{Ell}}(n)\longrightarrow \limfunc{%
Gauss}(n).  \label{bibi}
\end{equation}

\subsubsection{The non-centered case}

We now discuss the case of the set $\limfunc{Gauss}(n)$ of Gaussians with
arbitrary center $z_{0}=(x_{0},p_{0})$. Such functions are defined by%
\begin{equation*}
\psi _{z_{0},AB}=\widehat{T}(z_{0})\psi _{AB}
\end{equation*}%
where $\widehat{T}(z_{0})$ is the Heisenberg--Weyl displacement operator 
\begin{equation}
\widehat{T}(z_{0})\psi _{0}(x)=e^{\frac{i}{\hbar }(p_{0}\cdot x-\frac{1}{2}%
p_{0}\cdot x_{0})}\psi _{0}(x-x_{0}).  \label{hw1}
\end{equation}
Taking into account the trivial equality 
\begin{equation*}
p_{0}\cdot x-\frac{1}{2}p_{0}\cdot x_{0}=p_{0}\cdot (x-x_{0})+\frac{1}{2}%
p_{0}\cdot x_{0}
\end{equation*}%
these states are often written in the physical literature as $e^{\frac{i}{%
\hbar }p_{0}\cdot (x-x_{0})}\psi _{AB}$; this notation has the disadvantage
of making the symplectic covariance properties difficult to track. However
both choices lead to the same Wigner function 
\begin{equation*}
W\psi _{z_{0},AB}(z)=(\pi \hbar )^{-n}e^{-\tfrac{1}{\hbar }%
G_{AB}(z-z_{0})\cdot (z-z_{0})}.
\end{equation*}%
as follows from the translational property \cite{Birk,Wigner}%
\begin{equation*}
W(\widehat{T}(z_{0})\psi )(z)=W\psi \circ T(z_{0})^{-1}(z)=W\psi (z-z_{0}).
\end{equation*}

\subsection{Mixed geometric states}

\subsubsection{Mixed states, revisited}

Sofar we have been dealing with what would be called \textquotedblleft pure
states\textquotedblright\ in quantum mechanics. Such states are
traditionally represented by classes of non-zero single functions $\psi \in
L^{2}(\mathbb{R}^{n})$. For instance, this is the case of the Gaussian
functions $\psi _{A,B}^{\gamma }$ considered above. More generally, one
considers so-called mixed states: a mixed state is a countable family $(\psi
_{j},\alpha _{jj})_{j\in \mathcal{J}}$ where $||\psi _{j}||_{L^{2}(\mathbb{R}%
^{n})}=1$ and $\alpha _{j}\geq 0$, $\sum_{j\in \mathcal{J}}\alpha _{j}=1$.
The $\alpha _{jj}$ are viewed as probabilities. To the family $(\psi
_{j},\alpha _{j})_{j\in \mathcal{J}}$ one associates the density operator
(as defined in Section \ref{secdensity}) $\widehat{\rho }=\sum_{_{j\in 
\mathcal{J}}}\alpha _{j}\widehat{\rho }_{j}$ where $\widehat{\rho }_{j}$ is
the orthogonal projection in $L^{2}(\mathbb{R}^{n})$ on the ray $\mathbb{C}%
\psi _{j}$. At this point it might be useful to make the following remark:
physicists usually identify the mixed state with the density operator $%
\widehat{\rho }$ itself. But this leads to a slight inconsistency, because
several different mixed states (as we defined them above) can lead to the
same operator density $\widehat{\rho }$. We have discussed this problem in 
\cite{MCQR}.

Textbook examples of density operators representing mixed state are provided
by states with Wigner distribution 
\begin{equation}
\rho (z)=\left( \tfrac{1}{2\pi }\right) ^{n}(\det \Sigma )^{-1/2}e^{-\tfrac{1%
}{2}\Sigma ^{-1}z\cdot z}  \label{mixgauss}
\end{equation}%
where $\Sigma $ (the covariance matrix) satisfies the condition (\ref%
{quantum}), that is%
\begin{equation}
\Sigma +\frac{i\hbar }{2}J\geq 0.  \label{quantumbis}
\end{equation}%
As discussed above, this condition guarantees the positivity of the operator 
$\widehat{\rho }$ with Weyl symbol $(2\pi \hbar )^{n}\rho $ (see \cite%
{cogoni2}).

An usual measure of the \textquotedblleft mixedness\textquotedblright\ of a
Gaussian state is the \textit{purity} $\mu (\widehat{\rho })$; by definition 
$\mu (\widehat{\rho })=\limfunc{Tr}(\widehat{\rho }^{2})$ and we have $\mu (%
\widehat{\rho })=1$ if and only if $\widehat{\rho }$ represents a pure
state, \textit{i.e.} reduces to a rank one projector, in which case (\ref%
{mixgauss}) is just the Wigner function of a Gaussian (\ref{wpt}). We
mention that the purity of a Gaussian state (\ref{mixgauss}) with covariance
matrix $\Sigma $ is purity of a Gaussian state $\widehat{\rho }$ is \cite%
{Rodino} 
\begin{equation*}
\mu (\widehat{\rho })=\left( \frac{\hbar }{2}\right) ^{n}(\det \Sigma
)^{-1/2}
\end{equation*}%
hence $\mu (\widehat{\rho })=1$ if and only if $\det \Sigma =(\hbar /2)^{2n}$%
.

Let us define the notion of mixedness in the context of our geometric
quantum states.

\begin{definition}
Let $(\ell ,\ell ^{\prime })$ be a Lagrangian frame in $(\mathbb{R}%
^{2n},\sigma )$ and $X_{\ell }\subset \ell $ be an ellipsoid with center $0$%
. Let $P_{\ell ^{\prime }}\subset \ell ^{\prime }$ be an ellipsoid centered
at $0$ containing $(X_{\ell })_{\ell ^{\prime }}^{\hbar }$. If $P_{\ell
^{\prime }}\neq (X_{\ell })_{\ell ^{\prime }}^{\hbar }$ we call the product $%
X_{\ell }\times P_{\ell ^{\prime }}$ a geometric mixed quantum state in $%
\mathbb{R}^{2n}$ associated with the frame $(\ell ,\ell ^{\prime })$ and the
ellipsoids $X_{\ell }$ and $P_{\ell ^{\prime }}$.
\end{definition}

The following result connects this definition to our discussion above:

\begin{proposition}
Let $(\ell ,\ell ^{\prime })$ be a Lagrangian frame in $(\mathbb{R}%
^{2n},\sigma )$ and $X_{\ell }\times P_{\ell ^{\prime }}$ a geometric mixed
quantum state. The John ellipsoid $\Omega =(X_{\ell }\times P_{\ell ^{\prime
}})_{\mathrm{John}}$ is quantum admissible.
\end{proposition}

\begin{proof}
By definition we have $X\times (X_{\ell })_{\ell ^{\prime }}^{\hbar }\subset
X_{\ell }\times P_{\ell ^{\prime }}$ hence $(X_{\ell }\times P_{\ell
^{\prime }})_{\mathrm{John}}$ contains $(X_{\ell }\times (X_{\ell })_{\ell
^{\prime }}^{\hbar })_{\mathrm{John}}$. In view of Proposition \ref%
{PropStandard} we have 
\begin{equation*}
(X_{\ell }\times (X_{\ell })_{\ell ^{\prime }}^{\hbar })_{\mathrm{John}%
}=S(B^{2n}(\sqrt{\hbar }))
\end{equation*}%
for some $S\in \limfunc{Sp}(n)$ hence $(X_{\ell }\times P_{\ell ^{\prime
}})_{\mathrm{John}}$ is an admissible ellipsoid.
\end{proof}

The ellipsoid $\Omega $ in the result above plays the role of a statistical
covariance matrix (as it already did implicitly in the \textquotedblleft
pure case\textquotedblright\ studied in Section \ref{secpure}); in
particular condition (\ref{quantumbis}) is equivalent to saying that $\Omega 
$ contains a quantum blob $S(B^{2n}(\sqrt{\hbar }))$, $S\in \limfunc{Sp}(n)$.

\subsubsection{Symplectic capacities and geometric states}

Artstein-Avidan \textit{et al.} show in \cite{arkaos13} (Remark 4.2) that%
\textit{\ }if $X\subset \mathbb{R}_{x}^{n}$ and $P\subset \mathbb{R}_{p}^{n}$
are any centrally symmetric convey bodies, then we have 
\begin{equation}
c_{\max }(X\times P)=4\hbar \sup \{\lambda >0:\lambda X^{\hbar }\subset P\}
\label{yaron1}
\end{equation}%
hence, in particular,%
\begin{equation}
c_{\mathrm{HZ}\ }(X\times X^{\hbar })=c_{\max }(X\times X^{\hbar })=4\hbar .
\label{yaron2}
\end{equation}%
We can actually slightly improve (\ref{yaron2}) using our previous results.
We begin by remarking that when $\Omega \subset \mathbb{R}^{2n}$ is a
centrally symmetric body we have \cite{BSM} 
\begin{equation}
c_{\min }^{\mathrm{lin}}(\Omega )=\sup_{S\in \limfunc{Sp}(n)}\{\pi
R^{2}:S(B^{2n}(R))\subset \Omega \}~.  \label{clinmin}
\end{equation}

\begin{proposition}
Let $X_{\ell }\times (X_{\ell })_{\ell ^{\prime }}^{\hbar }\in \limfunc{Quant%
}\nolimits_{0}^{\mathrm{Ell}}(n)$. We have%
\begin{gather}
c_{\max }(X_{\ell }\times (X_{\ell })_{\ell ^{\prime }}^{\hbar })=c_{\mathrm{%
HZ}\ }(X_{\ell }\times (X_{\ell })_{\ell ^{\prime }}^{\hbar })=4\hbar
\label{nous1} \\
\text{\textit{and} \ }c_{\min }^{\mathrm{lin}}(X_{\ell }\times (X_{\ell
})_{\ell ^{\prime }}^{\hbar })~=4\hbar .  \label{nous2}
\end{gather}%
In the case of a mixed geometric quantum state $X_{\ell }\times P_{\ell
^{\prime }}$ we have 
\begin{equation}
c_{\max }(X_{\ell }\times P_{\ell ^{\prime }})=4\hbar \sup \{\lambda
>0:\lambda X_{\ell }^{\hbar }\subset P_{\ell ^{\prime }}\}\geq 4\hbar .
\label{nous3}
\end{equation}
\end{proposition}

\begin{proof}
Recall from Proposition \ref{PropStandard} that we have 
\begin{equation}
X_{\ell }\times (X_{\ell })_{\ell ^{\prime }}^{\hbar }=S(B_{X}^{n}(\sqrt{%
\hbar })\times B_{P}^{n}(\sqrt{\hbar }))  \label{xx}
\end{equation}%
for some $S\in \limfunc{Sp}(n)$. Formula (\ref{nous1}) follows using formula
(\ref{yaron2}) and the symplectic invariance of symplectic capacities.
Formula (\ref{nous2}) also follows from (\ref{xx}) using (\ref{clinmin}) and
Lemma \ref{LemmaJohn}. The formulas (\ref{nous3}) follow from (\ref{yaron1}).
\end{proof}

\section{Gaussian Beams and Geometric Quantum States}

\subsection{Symplectic and metaplectic isotopies}

Let $M(t)\in \limfunc{Sym}(2n,\mathbb{R})$ depend in a $C^{j}$ ($j\geq 2$)
fashion on $t\in \mathbb{R}$. To $M(t)$ we associate the time-dependent
quadratic Hamiltonian function 
\begin{equation}
H(z,t)=\frac{1}{2}M(t)z\cdot z.  \label{HM}
\end{equation}%
The flow determined by Hamilton's equations $\dot{z}=J\partial _{z}H(z,t)$
for $H$ is a linear symplectic isotopy $t\longmapsto S_{t}$ . Let $\widehat{H%
}$ be the Weyl quantization \cite{Birk,Birkbis} of the Hamiltonian function $%
H$; the exact solution of the associated Schr\"{o}dinger equation%
\begin{equation*}
i\hbar \frac{\partial \psi }{\partial t}=\widehat{H}\psi \text{ \ },\text{ \ 
}\psi (\cdot ,0)=\psi _{0}\in L^{2}(\mathbb{R}^{n})
\end{equation*}%
is given by $\psi (x,t)=\widehat{S}_{t}\psi _{0}(x)$ where $t\longmapsto 
\widehat{S}_{t}$ is the unique path of operators $\widehat{S}_{t}\in 
\limfunc{Mp}(n)$ such that $\pi ^{\limfunc{Mp}}(\widehat{S}_{t})=S_{t}$ and $%
\widehat{S}_{0}=\widehat{I}$ (see \cite{Birk,paths} and the references
therein). The existence of such a lifting follows from general properties of
covering groups, see for instance Steenrod \cite{Steenrod}.

\subsection{The method of Gaussian beams}

\subsubsection{The nearby orbit approximation}

We make the following assumptions on the time-dependent Hamiltonian function 
$H$: $H\in C^{\infty }(\mathbb{R}_{z}^{2n}\times \mathbb{R}_{t}\mathbb{)}$
and there exist constants $C_{\alpha }>0$, $T>0$ such that 
\begin{equation}
|\partial _{z}^{\alpha }H(z,t)|\leq C_{\alpha }\text{ \ \textit{for all} \ }%
|\alpha |\geq 2\text{ and }(z,t)\in \mathbb{R}^{2n}\times \lbrack -T,T].
\label{cond1}
\end{equation}%
These conditions \textit{a priori} exclude \textquotedblleft
physical\textquotedblright\ Hamiltonians of the type%
\begin{equation*}
H(x,p,t)=\frac{1}{2}|p|^{2}+V(x,t)
\end{equation*}%
because the latter are never bounded. But for all practical purposes we want
to study the Hamilton equations $\dot{x}=p$, $\dot{p}=-\partial _{x}V(x,t)$
in a \textit{bounded} domain $D$ of phase space. Choosing a compactly
supported cutoff function $\chi \in C_{0}^{\infty }(\mathbb{R}^{2n})$ such
that $\chi (z)=1$ for $z\in \overline{D}$ the function $\chi H$ \ satisfies
the conditions (\ref{cond1}) and the solutions $t\longmapsto (x,p)$ of the
Hamilton equations for $\chi H$ with initial value $z_{0}\in D$ are just
those of the initial problem $\dot{x}=p$, $\dot{p}=-\partial _{x}V(x,t)$ as
long as $(x,p)$ remains in $D$ (and the curves $t\longmapsto (x,p)$ stay
inside the support of $\chi $).

We next consider first order approximations to the solutions of Hamilton's
equations, which we obtain by replacing the full Hamiltonian function $H$
with its truncated second order Taylor expansion around $z_{t}$: 
\begin{equation}
H_{0}(z,t)=\partial _{z}H(z_{t},t)(z-z_{t})+\frac{1}{2}%
D_{z}^{2}H(z_{t},t)(z-z_{t})\cdot (z-z_{t})  \label{Hnought}
\end{equation}%
($D_{z}^{2}H$ is the Hessian matrix of $H$). Let us denote by $%
S_{t}^{H}(z_{0})$ the Jacobian matrix at $z_{0}$ of the symplectomorphism $%
f_{t}^{H}$: 
\begin{equation}
S_{t}^{H}(z_{0})=Df_{t}^{H}(z_{0})\in \limfunc{Sp}(n).  \label{so}
\end{equation}%
The linear symplectic isotopy $t\longmapsto S_{t}^{H}(z_{0})$ is the
solution of the differential equation \cite{AM,Birk}%
\begin{equation}
\frac{d}{dt}S_{t}^{H}(z_{0})=JD_{z}^{2}H(z_{t},t)S_{t}^{H}(z_{0})\text{ },%
\text{ \ }S_{0}^{H}(z_{0})=I_{2n\times 2n};  \label{variation}
\end{equation}%
it is sometimes called the \textquotedblleft variational
equation\textquotedblright\ in perturbation theory.

Using the uniqueness theorem for the solutions of Hamilton's equations it is
easy to shows the following approximation result:

\begin{proposition}
\label{PropBlob}Let $z_{t}=f_{t}^{H}(z_{0})$ and $z(t)=f_{t}^{H_{0}}(z(0))$
be the solutions of Hamilton's equations for $H$ and $H_{0}$, respectively.
(i) These solutions are related by the formula%
\begin{equation}
z(t)-z_{t}=S_{t}^{H}(z_{0})(z(0)-z_{0})  \label{bf}
\end{equation}%
where 
\begin{equation}
S_{t}^{H}(z_{0})=Df_{t}^{H}(z_{0})\in \limfunc{Sp}(n).  \label{stz0}
\end{equation}%
In particular $f_{t}^{H}(z_{0})=f_{t}^{H_{0}}(z_{0})$. (ii) Assume that 
\begin{equation*}
z(0)\in z_{0}+S(B^{2n}(\sqrt{\hbar }))
\end{equation*}
for some $S\in \limfunc{Sp}(n)$. Then 
\begin{equation*}
z(t)\in z_{t}+S_{t}^{H}(z_{0})S(B^{2n}(\sqrt{\hbar })).
\end{equation*}
\end{proposition}

\begin{proof}
(i) The Hamilton equations for $H_{0}$ are 
\begin{equation}
\dot{z}(t)=J\partial _{z}H(z_{t},t)+JD_{z}^{2}H(z_{t},t)(z(t)-z_{t}).
\label{hamk}
\end{equation}%
Setting $u(t)=z(t)-z_{t}$ we thus have 
\begin{equation*}
\dot{u}(t)+\dot{z}_{t}=J\partial _{z}H(z(t),t)+JD_{z}^{2}H(z_{t},t)u(t)
\end{equation*}%
that is, since $\dot{z}_{t}=J\partial _{z}H(z_{t},t)$,%
\begin{equation*}
\dot{u}(t)=JD_{z}^{2}H(z_{t},t)u(t).
\end{equation*}%
It follows from (\ref{variation}) that $u(t)=S_{t}(z_{0})(u(0))$ and hence%
\begin{equation*}
z(t)-z_{t}=S_{t}^{H}(z_{0})(z(0)-z_{0})
\end{equation*}%
which is (\ref{bf}). (ii) Assume that 
\begin{equation*}
z(0)\in z_{0}+S(B^{2n}(\sqrt{\hbar })
\end{equation*}%
for some $S\in \limfunc{Sp}(n)$; then by formula (\ref{bf})%
\begin{equation*}
z(t)-z_{t}=S_{t}^{H}(z_{0})(z(0)-z_{0})\in S_{t}^{H}(z_{0})S(B^{2n}(\sqrt{%
\hbar }).
\end{equation*}
\end{proof}

The statement (ii) above is to be contrasted with the following classical
estimate for the accuracy of the approximate trajectories. Using Gr\"{o}%
nwall's inequality \cite{jeannot} in integral form one shows that there
exists a constant $k>0$ such that 
\begin{equation}
|z(t)-z_{t}|\leq e^{k|t|}|z(0)-z_{0}|\text{ \ for }-T\leq t\leq T.
\label{estekt}
\end{equation}%
The proof of this inequality does not make use of the fact that $%
S_{t^{\prime }}^{H}(z_{0})$ is symplectic, and the result is actually a
consequence of the elementary theory of approximations to solutions of
differential equations (see for instance \cite{jeannot}). More striking (and
useful in our context) is the simple qualitative statement (ii) in
Proposition \ref{PropBlob} above which shows that if two initial phase space
points are contained in a quantum blob, they remain in a quantum blob during
time evolution.

\subsubsection{First order Gaussian beams}

Let $\widehat{H}=\limfunc{Op}^{\mathrm{Weyl}}(H(\cdot ,t))$ be the Weyl
quantization \cite{Birk,Birkbis} of the Hamiltonian function $H$, and
consider the Cauchy problem for the Schr\"{o}dinger equation 
\begin{equation}
i\hbar \partial _{t}\psi (x,t)=\widehat{H}\psi (x,t)\text{ \ , \ }\psi
(\cdot ,0)=\phi _{z_{0}}^{\hbar }  \label{Eqn1}
\end{equation}%
where $\phi _{z_{0}}^{\hbar }$ is the displaced standard Gaussian defined by 
\begin{equation}
\phi _{z_{0}}^{\hbar }(x)=\widehat{T}(z_{0})\phi _{0}^{\hbar }(x)
\label{fiz}
\end{equation}%
(recall that $\widehat{T}(z_{0})$ is the Heisenberg--Weyl displacement
operator (\ref{hw1})).

Let us introduce the following notation:

\begin{itemize}
\item \textit{The function} $t\longmapsto z_{t}$ \textit{is the solution of
Hamilton's equations} $\dot{z}=J\partial _{z}H(z,t)$ \textit{with initial
condition }$z_{0}$\textit{\ at time} $t=0$;

\item \textit{The symmetrized phase along the trajectory} $t\longmapsto
z_{t} $ \textit{is the real number} 
\begin{equation}
\gamma ^{H}(z_{0},t)=\int_{0}^{t}\left( \tfrac{1}{2}\sigma (z_{s},\dot{z}%
_{s})-H(z_{s},s)\right) ds~;  \label{phase}
\end{equation}

\item \textit{The metaplectic isotopy} $t\longmapsto \widehat{S}%
_{t}^{H}(z_{0})$ \textit{is the lift to} $\limfunc{Mp}(n)$ \textit{of the
symplectic isotopy} $t\longmapsto S_{t}^{H}(z_{0})\in \limfunc{Sp}(n)$ 
\textit{solution of the variational equation }(\ref{variation}).
\end{itemize}

One shows \cite{Berra,Schulze} the following approximation result which
historically goes back to Hagedorn \cite{Hagedorn} (many variants thereof
can be found in the literature):

\begin{proposition}
\label{PropBeam}Let $\psi $ be a solution of Schr\"{o}dinger's equation for $%
H$ with initial condition $\psi _{0}=\phi _{z_{0}}^{\hbar }$. The function $%
\psi _{z_{0}}(\cdot ,t)=\widehat{U}^{H}(z_{0},t)\phi _{z_{0}}^{\hbar }$
defined by 
\begin{equation}
\psi _{z_{0}}(x,t)=e^{\frac{i}{\hbar }\gamma ^{H}(z_{0},t)}\widehat{T}(z_{t})%
\widehat{S}_{t}^{H}(z_{0})\phi _{0}^{\hbar }(x)  \label{beam00}
\end{equation}%
satisfies the estimate%
\begin{equation}
||\psi (\cdot ,t)-\psi _{z_{0}}(\cdot ,t)||_{L^{2}}\leq C_{N}(z_{0})\hbar
^{1/2}|t|\text{ \ \textit{for} \ }0\leq t\leq T.  \label{pschitt1}
\end{equation}
\end{proposition}

We will call the (nonlinear) operator $\widehat{U}^{H}(z_{0},t)$ the \textit{%
Gaussian beam} associated with $H$ and the initial point $z_{0}$. Formula (%
\ref{beam00}) shows that the solution of Schr\"{o}dinger's equation with
initial datum $\phi _{0}^{\hbar }$ is approximated by the Gaussian obtained
by propagating the initial function $\phi _{0}^{\hbar }$ along the exact
Hamiltonian trajectory $t\longmapsto z_{t}$ starting from $z=0$ while
deforming it using the metaplectic lift of the linearized flow around this
point \cite{ACHA}.

Note that since $\phi _{z_{0}}^{\hbar }=\widehat{T}(z_{0})\phi _{0}^{\hbar }$
we can rewrite the approximate solution (\ref{beam00}) as 
\begin{equation}
\psi _{z_{0}}(x,t)=e^{\frac{i}{\hbar }\gamma ^{H}(z_{0},t)}\widehat{T}(z_{t})%
\widehat{S}_{t}^{H}(z_{0})\widehat{T}(-z_{0})\phi _{z_{0}}^{\hbar }(x)
\label{beam0}
\end{equation}%
We \ mention that in \cite{Berra} Berra \textit{et al}. generalize (\ref%
{beam0}) to obtain approximations arbitrary order $\mathcal{O}(\hbar ^{N})$.)

The following generalization of Proposition \ref{PropBeam} shows that the
estimate (\ref{pschitt1}) still holds when the initial condition $\psi
_{0}=\phi _{z_{0}}^{\hbar }$ is replaced with the more general condition $%
\psi _{0}=\psi _{A,B}$.

\begin{proposition}
Let $\psi $ be a solution of Schr\"{o}dinger's equation for $H$ with initial
condition $\psi _{0}=\widehat{T}(z_{0})\psi _{A,B}$. The function $\psi
_{AB;z_{0}}(\cdot ,t)=\widehat{U}_{AB}^{H}(z_{0},t)\psi _{0}$ defined by 
\begin{equation}
\widehat{U}^{H}(z_{0},t)\psi _{0}=e^{\frac{i}{\hbar }\gamma ^{H}(z_{0},t)}%
\widehat{T}(z_{t})\widehat{S}_{t}^{H}(z_{0})\psi _{A,B}\text{ }
\label{beam000}
\end{equation}%
satisfies the estimate (\ref{pschitt1}), that is,%
\begin{equation}
||\psi (\cdot ,t)-\psi _{z_{0}}(\cdot ,t)||_{L^{2}}\leq C_{N}(z_{0})\hbar
^{1/2}|t|\text{ \ \textit{for} \ }0\leq t\leq T.  \label{pschitt2}
\end{equation}
\end{proposition}

\begin{proof}
We have $\psi _{A,B}=\widehat{S}_{AB}\phi _{0}^{\hbar }$ where $\widehat{S}%
_{AB}=\widehat{V}_{B}\widehat{M}_{A^{1/2},0}$ is in $\limfunc{Mp}(n)$. We
can thus rewrite formula (\ref{beam00}) as 
\begin{eqnarray*}
\psi _{z_{0}}(x,t) &=&e^{\frac{i}{\hbar }\gamma ^{H}(z_{0},t)}\widehat{T}%
(z_{t})\widehat{S}_{AB}\left( \widehat{S}_{AB}^{-1}\widehat{S}_{t}^{H}(z_{0})%
\widehat{S}_{AB}\right) \phi _{0}^{\hbar } \\
&=&\widehat{S}_{AB}\left[ e^{\frac{i}{\hbar }\gamma ^{H}(z_{0},t)}\widehat{T}%
(S_{AB}^{-1}z_{t})\left( \widehat{S}_{AB}^{-1}\widehat{S}_{t}^{H}(z_{0})%
\widehat{S}_{AB}\right) \phi _{0}^{\hbar }\right] 
\end{eqnarray*}%
and hence 
\begin{equation*}
\widehat{S}_{AB}^{-1}\psi _{z_{0}}(x,t)=e^{\frac{i}{\hbar }\gamma
^{H}(z_{0},t)}\widehat{T}(S_{AB}^{-1}z_{t})\left( \widehat{S}_{AB}^{-1}%
\widehat{S}_{t}^{H}(z_{0})\widehat{S}_{AB}\right) \phi _{0}^{\hbar }
\end{equation*}%
We next note the following facts: (a) if $\psi $ is a solution of Schr\"{o}%
dinger's equation for $H$ then $\widehat{S}_{AB}^{-1}\psi $ is a solution of
Schr\"{o}dinger's equation for $H_{AB}=H\circ S_{AB}$ where $S_{AB}=\pi ^{%
\limfunc{Mp}}(\widehat{S}_{AB})$. This follows from the symplectic
conjugation relation $\widehat{H_{AB}}=\widehat{S}_{AB}^{-1}\widehat{H}%
\widehat{S}_{AB}$ (properties of Weyl pseudo-differential calculus, see \cite%
{Birkbis}, p. 254). Similarly $t\longmapsto S_{AB}^{-1}z_{t}$ is the
solution to Hamilton's equation for $H_{AB}$; (b) we have $\gamma
^{H}(z_{0},t)=\gamma ^{H_{AB}}(z_{0},t)$ since replacing $H$ with $H_{AB}$
leads to the replacement of $z_{t}$ with $S_{AB}z_{t}$; (b) a
straightforward calculation of Hessian matrices shows that%
\begin{equation*}
\widehat{S}_{AB}^{-1}\widehat{S}_{t}^{H}(z_{0})\widehat{S}_{AB}=\widehat{%
S_{t}}^{H_{AB}}(z_{0}).
\end{equation*}%
From these facts follows that we have 
\begin{equation*}
\widehat{S}_{AB}^{-1}\psi _{z_{0}}(x,t)=e^{\frac{i}{\hbar }\gamma
^{H_{AB}}(z_{0},t)}\widehat{T}(S_{AB}^{-1}z_{t})\widehat{S_{t}}%
^{H_{AB}}(z_{0})\phi _{0}^{\hbar }.
\end{equation*}%
Now, the Hamiltonian function $H_{AB}$ satisfies the same estimates (\ref%
{cond1}) as $H$;.we may therefore apply Proposition \ref{PropBeam} to $%
\widehat{S}_{AB}^{-1}\psi _{z_{0}z_{0}}$ and $\widehat{S}_{AB}^{-1}\psi $%
,which yields the estimate 
\begin{equation*}
||\widehat{S}_{AB}^{-1}\psi (\cdot ,t)-\widehat{S}_{AB}^{-1}\psi
_{z_{0}}(\cdot ,t)||_{L^{2}}\leq C_{N}(z_{0})\hbar ^{1/2}|t|\text{ \ \textit{%
for} \ }0\leq t\leq T.
\end{equation*}%
Since $\widehat{S}_{AB}\in \limfunc{Mp}(n)$ is unitary we get (\ref{pschitt2}%
).
\end{proof}

\subsection{Action of Gaussian beams on $\limfunc{Quant}\nolimits^{\mathrm{%
Ell}}(n)$}

Let $X_{\ell }\times (X_{\ell })_{\ell ^{\prime }}^{\hbar }\in \limfunc{Quant%
}\nolimits_{0}^{\mathrm{Ell}}(n)$ be a centered geometric quantum state. We
recall that the symplectic group $\limfunc{Sp}(n)$ acts on $\limfunc{Quant}%
\nolimits_{0}^{\mathrm{Ell}}(n)$ via the law (\ref{sxlrule}). If $%
t\longmapsto S_{t}$ is the linear symplectic isotopy generated by a
quadratic Hamiltonian function (\ref{HM}) we can define the action of this
symplectic isotopy on $X_{\ell }\times (X_{\ell })_{\ell ^{\prime }}^{\hbar }
$ by 
\begin{equation}
S_{t}(X_{\ell }\times (X_{\ell })_{\ell ^{\prime }}^{\hbar })=(S_{t}X)_{\ell
_{t}}\times \left[ S_{t}(X_{\ell })\right] _{\ell _{t}^{\prime }}^{\hbar })
\label{xslh}
\end{equation}%
where we have set $\ell _{t}=S_{t}\ell $ and $\ell _{t}^{\prime }=S_{t}\ell
^{\prime }$. In particular, if we consider the canonical geometric state (%
\ref{fidu}) we have%
\begin{equation}
S_{t}(X_{\ell _{X}}\times (X_{\ell _{X}})_{\ell _{P}}^{\hbar
})=S_{t}(B_{X}^{n}(\sqrt{\hbar })\times B_{P}^{n}(\sqrt{\hbar }))  \label{xb}
\end{equation}%
and to the latter corresponds, via Theorem \ref{Thm1} the Gaussian function $%
\psi =\widehat{S_{t}}\phi _{0}^{\hbar }$ which is the solution of Schr\"{o}%
dinger's equation%
\begin{equation*}
i\hbar \partial _{t}\psi =\widehat{H}\psi \text{ \ , \ }\psi (\cdot ,0)=\phi
_{0}^{\hbar }.
\end{equation*}%
the bijection More generally, we have the following property which shows
that Gaussian beams take geometric states to geometric states (we recall tat 
$\func{ISp}(n)=\limfunc{Sp}(n)\ltimes \mathbb{R}^{2n}$):

\begin{theorem}
\label{ThmBeam}Let $H\in C^{j}(\mathbb{R}^{2n}\times \mathbb{R})$, $j\geq 2$%
, be a Hamilton function satisfying the bounds (\ref{cond1}) and $\Psi $ the
extension of the bijection $\limfunc{Quant}\nolimits^{\mathrm{Ell}%
}(n)\longrightarrow \limfunc{Gauss}(n)$ defined by (\ref{gausson}) and (\ref%
{bibi}). We    
\begin{equation*}
X_{\ell (z_{0})}\times (X_{\ell (z_{0})})_{\ell ^{\prime }(z_{0})}^{\hslash
}\in \limfunc{Quant}\nolimits^{\mathrm{Ell}}(n).
\end{equation*}%
Let $U^{H}(z_{0},t)\in \func{ISp}(n)$ be defined by%
\begin{equation}
U^{H}(z_{0},t)=T(z_{t})S_{t}^{H}(z_{0})\text{.}  \label{uth}
\end{equation}%
We have the intertwining relation%
\begin{equation*}
\Psi \circ U^{H}(z_{0},t)=\widehat{U}^{H}(z_{0},t)\circ \Psi .
\end{equation*}
\end{theorem}

\begin{proof}
It immediately follows from the discussion above using the bijection (\ref%
{bibi}).
\end{proof}

\begin{acknowledgement}
This work has been financed by the Grant P 33447 N of the Austrian Research
Foundation FWF.
\end{acknowledgement}


\begin{thebibliography}{99}
\bibitem{AM} R. Abraham and J. E. Marsden. \textit{Foundations of Mechanics}%
. Second Edition, revised, enlarged, and reset, Addison--Wesley Publishing
Company, Inc., Redwood City, CA, 1987

\bibitem{Andrei} N. Andrei. Convex functions. \textit{Advanced Modeling and
Optimization} 9(2), 257--267 (2007)

\bibitem{arkaos13} S. Artstein-Avidan, R. Karasev, and Y. Ostrover. From
Symplectic Measurements to the Mahler Conjecture. \textit{Duke Math. J.}
163(11), 2003--2022 (2014)

\bibitem{Mellipsoid} S. Artstein-Avidan, V. Milman, and Y. Ostrover. The
M-ellipsoid, symplectic capacities and volume, \textit{Comment. Math. Helv}.
83, 359--369 (2008)

\bibitem{Arvind} Arvind, B. Dutta, N. Mukunda, and R. Simon. The real
symplectic groups in quantum mechanics and optics, \textit{Pramana Journal
of Physics}, 45(6), 471--497 (1995)

\bibitem{ABMB} G. Aubrun and S. J. Szarek. \textit{Alice and Bob meet Banach}%
. American Mathematical Soc. Vol. 223, 2017

\bibitem{Ball} K. M. Ball. Ellipsoids of maximal volume in convex bodies. 
\textit{Geom. Dedicata.} 41(2), 241--250 (1992)

\bibitem{Banyaga} A. Banyaga. Sur la structure du groupe des diff\'{e}%
omorphismes qui pr\'{e}servent une forme symplectique, \textit{Comm. Math.
Helv.} 53, 174--227 (1978)

\bibitem{Berra} M. Berra, I. M. Bulai, E. Cordero,, and F. Nicola. Gabor
frames of Gaussian beams for the Schr\"{o}dinger equation, \textit{Appl.
Comput.Harmon.Anal}.43, 94--121 (2017)

\bibitem{Bianchi} G. Bianchi and M. Kelly. A Fourier analytic proof of the
Blaschke--Santal\'{o} inequality. \textit{Proc. Am. Math. Soc}. 143(11),
1901--4912 (2015)

\bibitem{Blaschke} W. Blaschke. \"{U}ber affine Geometrie VII: Neue
Extremeingenschaften von Ellipse und Ellipsoid, \textit{Ber. Verh. S%
%TCIMACRO{\U{a8}}%
%BeginExpansion
\"{}%
%EndExpansion
achs. Akad. Wiss., Math. Phys.} Kl. 69, 412--420 (1917)

\bibitem{BM} J. Bourgain and V. Milman. New volume ratio properties for
convex symmetric bodies, \textit{Invent. Math. }88, 319--340 (1987)

\bibitem{Boyd} S. Boyd, S. P. Boyd, and L. Vandenberghe. \textit{Convex
optimization}. Cambridge university press, 2004

\bibitem{cielibak} K. Cieliebak, H. Hofer, Latschev, F. Schlenk.
Quantitative symplectic geometry. arXiv preprint math/0506191 (2005)

\bibitem{cogoni2} E. Cordero, M. de Gosson, and F. Nicola. On the positivity
of trace class operators. \textit{Adv. Theor. Math. Phys.} 23(8), 2061--2091
(2019)

\bibitem{jeannot} J. Dieudonn\'{e}. \textit{Foundations of Modern Analysis}.
New York, Academic Press, 1960.

\bibitem{ekhof1} I. Ekeland and H. Hofer. Symplectic topology and
Hamiltonian dynamics, \textit{Math. Z.} 200(3), 355--378 (1989)

\bibitem{ekhof2} I. Ekeland and H. Hofer. Symplectic topology and
Hamiltonian dynamics, \textit{Math. Z.} 203, 553--567 (1990)

\bibitem{fe06} H. G. Feichtinger. Modulation Spaces: Looking Back and Ahead. 
\textit{Sample} \textit{Theory Signal Image Process}, 5(2), 109--140 (2006)

\bibitem{physletta} M. de Gosson. Phase Space Quantization and the
Uncertainty Principle. \textit{Phys. Lett. A} 317(5--6) (2003)

\bibitem{CM} C. de Gosson and M. de Gosson. On the Non-Uniqueness of
Statistical Ensembles Defining a Density Operator and a Class of Mixed
Quantum States with Integrable Wigner Distribution. \textit{Quantum Rep.}
3(3), 473--481 (2021)

\bibitem{ACHA} M. de Gosson. Hamiltonian deformations of Gabor frames: First
steps, A\textit{ppl. Comput. Harmon.\ Anal}. 38(2) 195--223 (2015)

\bibitem{paths} M. de Gosson, Paths of Canonical Transformations and their
Quantization. \textit{Rev. Math. Phys. }27(6)\textit{, }1530003 (2015)

\bibitem{FOOP} M. de Gosson and C. de Gosson. Pointillisme \`{a} la Signac
and Construction of a Quantum Fiber Bundle Over Convex Bodies. \textit{%
Found. Phys. }53(43) (2023)

\bibitem{Birk} M. de Gosson, \textit{Symplectic geometry and quantum
mechanics.} Vol. 166. Springer Science \& Business Media, 2006

\bibitem{go09} M. de Gosson. The Symplectic Camel and the Uncertainty
Principle: The Tip of an Iceberg? \textit{Found. Phys}. 99, 194 (2009)

\bibitem{blob} M. de Gosson. Quantum blobs. \textit{Found. Phys.} 43(4),
440--457\textbf{\ }(2013)

\bibitem{Birkbis} M. de Gosson. \textit{Symplectic Methods in Harmonic
Analysis and in Mathematical Physics}, Birkh\"{a}user, 2011

\bibitem{acha} M. de Gosson, Two Geometric Interpretations of the
Multidimensional Hardy Uncertainty Principle. \textit{Appl. Comput. Harmon.
Anal.} 42(1), 143--153\textit{\ }(2017)

\bibitem{Wigner} M. de Gosson, \textit{The Wigner Transform, World Scientific%
}, series Advanced Texts in Mathematics, 2017

\bibitem{Rodino} M. de Gosson. On density operators with Gaussian Weyl
symbols. (English) Zbl 07218925Boggiatto, Paolo (ed.) et al., Advances in
microlocal and time-frequency analysis. MLTFA18, in honor of Prof. Luigi
Rodino on the occasion of his 70th birthday, Birkh\"{a}user. Appl. Numer.
Harmon. Anal., 191--206 (2020)

\bibitem{ATMP} M. de Gosson. Symplectic coarse-grained dynamics: Chalkboard
motion in classical and quantum mechanics. \textit{Adv. Theor. Math. Phys.}
24(4), 925--977 (2020)

\bibitem{QHA} M. de Gosson. \textit{Quantum Harmonic Analysis, an
Introduction}, De Gruyter, 2021

\bibitem{gopolar} M. de Gosson. Quantum Polar Duality and the Symplectic
Camel: a New Geometric Approach to Quantization. \textit{Found. Phys}. 51,
Article number: 60 (2021)

\bibitem{BSM} M. de Gosson. Polar Duality Between Pairs of Transversal
Lagrangian Planes; Applications to Uncertainty Principles.\textit{\ Bull.
sci. math} 179, 103171 (2022)

\bibitem{MCQR} C. de Gosson and M. de Gosson. On the Non-Uniqueness of
Statistical Ensembles Defining a Density Operator and a Class of Mixed
Quantum States with Integrable Wigner Distribution. \textit{Quantum Reports}
3(3), 473--48 (2021)

\bibitem{MCFOOP} M. de Gosson and C. de Gosson. Pointillisme \`{a} la Signac
and Construction of a Quantum Fiber Bundle Over Convex Bodies. \textit{%
Found. of Phys.} 53(2), paper n0. 43 (2023)

\bibitem{golulett} M. de Gosson and F. Luef. Quantum States and Hardy's
Formulation of the Uncertainty Principle: a Symplectic Approach. \textit{%
Lett. Math. Phys}. 80, 69--82 (2007)

\bibitem{golu09} M. de Gosson and F. Luef. Symplectic Capacities and the
Geometry of Uncertainty: the Irruption of Symplectic Topology in Classical
and Quantum Mechanics. \textit{Phys. Reps.} 484, 131--179 (2009)

\bibitem{Gro} K. Gr\"{o}chenig. \textit{Foundations of Time-Frequency
Analysis}, Birkh\"{a}user, Boston, 2000

\bibitem{gr85} M. Gromov. Pseudoholomorphic curves in symplectic manifolds. 
\textit{Inv. Math.} 82(2), 307--347 (1985)

\bibitem{GS} V. Guillemin and S. Sternberg. \textit{Geometric asymptotics}.
No. 14. American Mathematical Soc., 1990

\bibitem{Hagedorn} A. Hagedorn. Semiclassical quantum mechanics. I. The $%
\hbar \rightarrow 0$ limit for coherent states. \textit{Commun. Math. Phys.}
71(1):77--93 (1980)

\bibitem{Hardy} G. H. Hardy, A theorem concerning Fourier transforms. J%
\textit{. London Math. Soc}. 8, 227--231(1933)

\bibitem{HZ} H. Hofer and E. Zehnder. \textit{Symplectic Invariants and
Hamiltonian Dynamics}, Birkh\"{a}user Advanced Texts (Basler Lehrb\"{u}%
cher), Birkh\"{a}user Verlag, 1994

\bibitem{Jakobsen} M. S. Jakobsen, On a (no longer) New Segal Algebra: a
review of the Feichtinger algebra, J\textit{. Fourier Anal. Appl.} 24(6),
1579--1660 (2018)

\bibitem{Kaiser} M. J. Kaiser. The mixed volume optimization problem. 
\textit{Comput. Geom}. 12, 177--217 (1999)

\bibitem{Kastler} D. Kastler. The $C^{\ast }$-Algebras of a Free Boson
Field, \textit{Commun. math. Phys}. 1, 14--48 (1965)

\bibitem{Kuper} G. Kuperberg. From the Mahler Conjecture to Gauss Linking
Integrals, \textit{Geom. Funct. Anal.} 18(3), 870--892 (2008)

\bibitem{Mahler} K. Mahler. Ein \"{U}bertragungsprinzip f\"{u}r konvexe K%
\"{o}rper. \textit{\v{C}asopis pro p\v{e}stov\'{a}n\'{\i} matematiky a fysiky%
} 68(3), 93--102 (1939)

\bibitem{Messiah} A. Messiah. \textit{Quantum mechanics}, Vol. 1 (North
Holland, Amsterdam, 1970) translated by GM Temmer

\bibitem{Meyer} M. Meyer, C. Sch\"{u}tt, and E. M. Werner. New affine
measures of symmetry for convex bodies. \textit{Adv. Math. (NY)} 228,
2920--2942 (2011)

\bibitem{mepa} M. Meyer and A. Pajor. On the Blaschke--Santal\'{o}
inequality. \textit{Arch. Math}. 55, 82--93 (1990)

\bibitem{MeyerWerner} M. Meyer and E. Werner. The Santal\'{o}-regions of a
convex body. \textit{Trans. Am. Math. Soc}. 350(11), 4569--4591 (1998)

\bibitem{Milman} V. D. Milman. Geometrization of probability. Progress
Mathematics-Boston 265, 647 (2008)

\bibitem{Narcow} F. J. Narcowich. Conditions for the convolution of two
Wigner functions to be itself a Wigner function, \textit{J. Math. Phys}.
30(11), 2036--2041 (1988)

\bibitem{Schulze} V. Nazaikiinskii, B. W. Schulze, and B. Sternin. \textit{%
Quantization Methods in Differential Equations} (London: Taylor and Francis)
2002

\bibitem{Leonid} L. Polterovich. \textit{The geometry of the group of
symplectic diffeomorphisms}. Birkh\"{a}user, 2012

\bibitem{Santalo} L. A. Santal\'{o}. Un invariante a n para los cuerpos
convexos del espacio de $n$ dimensiones. \textit{Portugaliae. Math}. 8,
155--161 (1949)

\bibitem{sch} R. Schneider: \textit{Convex bodies: The Brunn-Minkowski theory%
}. Encyclopedia of Math. and its Applic.44, Cambridge University Press, 1993

\bibitem{sh87} M. A. Shubin, \textit{Pseudodifferential Operators and
Spectral Theory}, Springer-Verlag,. 1987 [original Russian edition in Nauka,
Moskva. 1978]

\bibitem{son} N. T. Son, P.-A. Absil, B. Gao, and T. Stykel. Computing
symplectic eigenpairs of symmetric positive-definite matrices via trace
minimization and Riemannian optimization, SIAM \textit{J. Matrix Anal. Appl.}
42 1732--1757 (2021)

\bibitem{Steenrod} N. Steenrod. \textit{The Topology of Fibre Bundles}.
Princeton Mathematical Series. Vol. 14.Princeton, N.J.: PUP, 1999

\bibitem{Vershynin} R. Vershynin. Lectures in Geometric Functional Analysis.
Unpublished manuscript. Available at http://www-personal. umich.
edu/romanv/papers/GFA-book/GFA-book. pdf 3.3 (2011): 3-3

\bibitem{Williamson} J. Williamson. On the algebraic problem concerning the
normal forms of linear dynamical systems. \textit{Am. J. Math.},
58(1):141--163 (1936)

\bibitem{zehnder} E. Zehnder, \textit{Lectures on Dynamical Systems:
Hamiltonian Vector Fields and Symplectic Capacities}. EMS Textbooks in
Mathematics. European Mathematical Society, 2010.
\end{thebibliography}
\end{document}